\documentclass[10pt,onecolumn,draftclsnofoot,journal]{IEEEtran}

\IEEEoverridecommandlockouts                              

\usepackage{amsmath,graphicx,multicol}
\usepackage{amsfonts}
\usepackage{color}
\usepackage{bbm}
\usepackage{cite}
\usepackage{amssymb}
\usepackage{algorithm}
\usepackage{algorithmic}
\usepackage{tabularx}
\usepackage[utf8]{inputenc}
\usepackage[english]{babel}

\usepackage{amsthm}
\usepackage{subcaption}
\usepackage[font={footnotesize}]{caption}
\usepackage{url}
\usepackage{mathtools}
\usepackage{comment}
\usepackage[keeplastbox]{flushend}
\newtheorem{definition}{Definition}[]
\newtheorem{proposition}{Proposition}[]
\newtheorem{theorem}{Theorem}[]
\newtheorem{corollary}{Corollary}[theorem]
\newtheorem{lemma}[]{Lemma}
\newcommand\norm[1]{\left\lVert#1\right\rVert}
\newcommand{\argmax}{\arg\!\max}

\DeclareMathOperator{\E}{\mathbb{E}}

\def\x{{\mathbf x}}
\def\z{{\mathbf z}}
\def\y{{\mathbf y}}

\def\B{{\mathbf B}}
\def\H{{\mathbf H}}

\def\P{{\mathbf P}}

\def\R{{\mathbb{R}}}

\def\I{{\mathbf I}}

\def\F{{\mathbf F}}

\def\x{{\mathbf x}}
\def\w{{\mathbf w}}

\def\X{{\mathbf X}}

\def\Y{{\mathbf Y}}

\def\B{{\mathbf B}}
\def\P{{\mathbf P}}

\def\R{{\mathbb{R}}}

\def\I{{\mathbf I}}

\def\h{{\mathbf h}}

\def\v{{\mathbf v}}

\newcommand{\ts}{\textsuperscript}
\newcommand{\abs}[1]{\lvert #1 \rvert}

\newcommand{\etal}{{\em et~al.~}}
\newcommand{\Var}{\mathrm{Var}}

\let\oldref\ref
\renewcommand{\ref}[1]{(\oldref{#1})}
\newcommand{\RNum}[1]{\uppercase\expandafter{\romannumeral #1\relax}}
\makeatletter
\renewcommand{\fnum@figure}{Fig.~\thefigure}
\makeatother

\title{
Randomized Greedy Sensor Selection: Leveraging Weak Submodularity
}

\author{Abolfazl~Hashemi,~\IEEEmembership{Student Member,~IEEE,} Mahsa~Ghasemi,~\IEEEmembership{Student Member,~IEEE,} Haris~Vikalo,~\IEEEmembership{Senior Member,~IEEE,} and Ufuk~Topcu
\thanks{Abolfazl Hashemi and Haris Vikalo are with the Department of Electrical and Computer Engineering, 
University of Texas at Austin, Austin, TX 78712 USA. Mahsa Ghasemi is with the Department of Mechanical 
Engineering, University of Texas at Austin, Austin, TX 78712 USA. Ufuk Topcu is with the Department of Aerospace 
Engineering and Engineering Mechanics, University of Texas at Austin, Austin, TX 78712 USA. Parts of the 
results in the paper were presented at the American Control Conference, Milwaukee, Wisconsin, USA, 
June 2018 \cite{ma}.}%
}
\begin{document}
%
 \maketitle
\begin{abstract}
We study the problem of estimating a random process from the observations collected by a network 
of sensors that operate under resource constraints. When the dynamics of the process and sensor 
observations are described by a state-space model and the resource are unlimited, the conventional 
Kalman filter provides the minimum mean-square error (MMSE) estimates. However, 
at any given time, restrictions on the available communications bandwidth and computational capabilities 
and/or power impose a limitation on the number of network nodes whose observations can be used to 
compute the estimates. We formulate the problem of selecting the most informative subset of the sensors as 
a combinatorial problem of maximizing a monotone set function under a uniform matroid constraint. 
For the MMSE estimation criterion we show that the maximum element-wise 
curvature of the objective function satisfies a certain upper-bound constraint and is, therefore, weak
submodular. We develop an efficient randomized greedy algorithm for sensor selection and establish 
guarantees on the estimator's performance in this setting. Extensive simulation results demonstrate the
efficacy of the randomized greedy algorithm compared to state-of-the-art greedy and semidefinite 
programming relaxation methods.
\end{abstract}
\vspace{-0.2cm}
\section{Introduction}\label{sec:intro}
\IEEEPARstart{M}{odern} sensor networks deploy a large number of nodes that either exchange their
noisy and possibly processed observations of a random process or forward those observations to a data 
fusion center. Due to constraints on computation, power and communication resources, instead of 
estimating the process using information collected by the entire network, the fusion center typically 
queries a relatively small subset of the available sensors. The problem of selecting the sensors that 
would acquire the most informative observations arises in a number of 
applications in control and signal processing systems including sensor selection for Kalman filtering \cite{nordio2015sensor,shamaiah2010greedy,tzoumas2016sensor}, batch state and stochastic 
process estimation \cite{tzoumas2016scheduling,tzoumas2016near}, minimal actuator placement 
\cite{summers2016submodularity,tzoumas2015minimal}, voltage control and meter placement in 
power networks \cite{damavandi2015robust,gensollen2016submodular,liu2016towards}, sensor 
scheduling in wireless sensor networks \cite{shamaiah2012greedy,nordio2015sensor}, and subset 
selection in machine learning \cite{mirzasoleiman2014lazier}.

For a variety of performance criteria, finding an optimal subset of sensors requires solving a 
computationally challenging combinatorial optimization problem, possibly using branch-and-bound search 
\cite{welch1982branch}. By reducing it to the set cover problem, sensor selection was in fact
shown to be NP-hard \cite{williamson2011design}. This hardness result has motivated development of numerous heuristics 
and approximate algorithms. For instance, \cite{joshi2009sensor} formulated the sensor selection 
problem as the maximization (minimization) of the $\log \det$ of the Fisher information matrix (error 
covariance matrix), and found a solution by relaxing the problem to a semidefinite program (SDP). 
The computational complexity of finding the solution to the SDP relaxation of the sensor selection 
problem is cubic in the total number of available sensors, which limits its practical feasibility in large-scale networks consisting of many sensing nodes. Moreover, the solution to the SDP 
relaxation comes 
with no performance guarantees. To overcome these drawbacks, Shamaiah et al. 
\cite{shamaiah2010greedy} proposed a greedy algorithm guaranteed to achieve at least 
$(1-1\slash e)$ of the optimal objective at a complexity lower than that of the SDP relaxation. 
The theoretical underpinnings of the greedy approach to the sensor selection problem in 
\cite{shamaiah2010greedy} are drawn from the area of submodular function optimization.
In particular, these results stem from the fact that the logarithm of the determinant ($\log \det$)
of the Fisher information matrix is a monotone 
submodular function. Nemhauser et al. \cite{nemhauser1978analysis} studied maximization 
of such a function subject to a uniform matroid constraint and showed that the greedy 
algorithm, which iteratively selects items providing maximum marginal gain, achieves $(1-1/e)$
approximation factor. More recently, 
\cite{tzoumas2016scheduling,tzoumas2015minimal,tzoumas2016near,tzoumas2016sensor}, 
employed and analyzed greedy algorithms for finding approximate solutions to the $\log \det$ 
maximization problem in a number of practical settings. 

Most of the existing work on greedy sensor selection has focused on optimizing the $\log \det$ of the 
Fisher information matrix, an objective indicative of the volume of the confidence ellipsoid. 
However, this criterion does not explicitly relate to the mean-square error (MSE) which is often 
a natural performance metric of interest in estimation problems. 
The MSE, i.e., the trace of the covariance matrix of the estimation error, is not supermodular 
\cite{krause2008near,jawaid2015submodularity,zhang2015sensor,singh2017supermodulara,
singh2017supermodular,olshevsky2017non}. Therefore, its negative value, which we would
like to maximize, is not submodular. Consequently, the setting and results of 
\cite{nemhauser1978analysis} do not apply to the MSE minimization problem.
Recently, Wang \etal \cite{wang2016approximation} analyzed performance of the greedy algorithm 
in the general setting of maximizing a monotone non-decreasing objective function that is not necessarily 
submodular. They used a notion of the total curvature $\mu$ of the objective function to show that the 
greedy algorithm provides a $((1+\mu)^{-1})$-approximation under a matroid constraint. However, 
determining the elemental curvature defined in \cite{wang2016approximation} is itself an NP-hard problem. 
Therefore, providing performance guarantees for the settings where the objective function is not 
submodular or supermodular, such as the trace of the covariance matrix of the estimation error in the sensor 
selection problem, remains a challenge. 
\begin{figure}[t]
	\centering
	\includegraphics[width=0.49\textwidth]{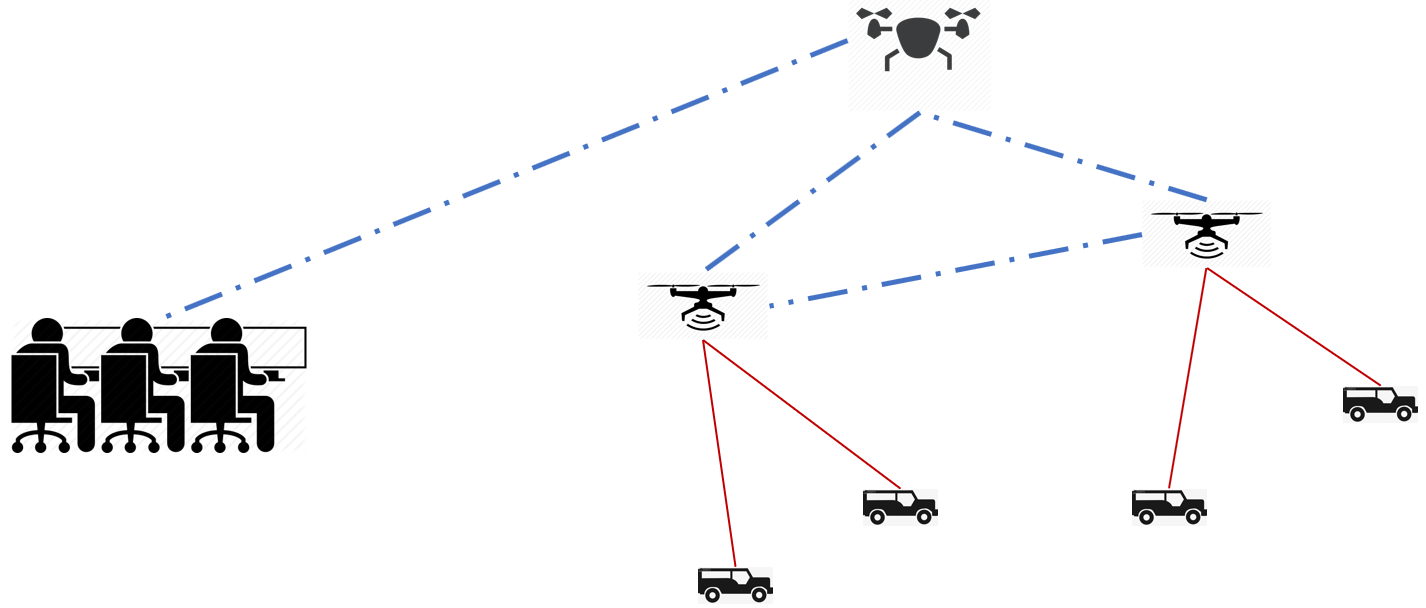}
	\caption{Multi-object tracking via a swarm of UAVs. The UAVs can communicate with each other 
	and are equipped with GPS and radar systems. The objective is to select a small subset of range 
	and angular measurements gathered by the UAVs to communicate to the control unit.}
	\label{fig:drone}
	\vspace{-0.3cm}
\end{figure}

On another note, processing massive amounts of data collected by modern large-scale networks may 
be challenging even for greedy algorithms. To further reduce the computational burden of 
maximizing a monotone increasing, submodular function subject to cardinality constraints, the authors of
\cite{mirzasoleiman2014lazier} proposed a stochastic greedy algorithm that achieves 
($1-1/e-\epsilon$)-approximation factor, where $\epsilon$ denotes a parameter that can be varied to 
explore the performance-complexity trade-off. However, the results of \cite{mirzasoleiman2014lazier} 
do not apply to the sensor selection problem under the (non-submodular) MSE objective.

In this paper, we formulate the task of selecting sensors in a large-scale network as the problem of 
maximizing a monotone non-submodular objective function that is directly related to the mean-square
estimation error. By analyzing curvature of the objective function, we derive sufficient conditions 
under which the function is weak submodular. In the important setting where a state-space model describes the dynamics of the process and sensor observations, we propose a randomized greedy 
algorithm and find a bound on the MSE of the state estimate formed by the Kalman filter that uses 
the measurements of the sensors selected by the proposed algorithm. A further implication 
of these results is that, when the measurement vectors are Gaussian or Bernoulli as frequently encountered in reduced-dimensionality Kalman filtering using random projections \cite{berberidis2016data},
the MSE objective is weak submodular with high probability. Our extensive simulations demonstrate 
that the proposed randomized greedy sensor selection scheme significantly outperforms both greedy 
and SDP relaxation methods in terms of computational complexity, and hence runtime, while providing
essentially the same or improved MSE performance. 
 
The rest of the paper is organized as follows. Section \oldref{sec:sys} presents a motivating example
and sets up the system model. In 
Section \oldref{sec:alg} we describe the novel formulation of the sensor selection problem and 
derive a bound on the curvature of the MSE-related objective function. In Section \oldref{sec:anl}, we 
introduce the randomized greedy algorithm and analyze its performance. Section \oldref{sec:sim} 
presents the simulation results while Section \oldref{sec:concl} states the concluding remarks. 

\textbf{Notation:} 
Bold capital letters denote matrices while bold lowercase letters represent vectors. 
$H_k(i,j)$ is the $(i,j)$ entry of the time-varying matrix $\H_k$ at time $k$, $\h_{k,j}$ 
is the $j\ts{th}$ row of $\H_k$, $\H_{k,S}$ is a submatrix of $\H_k$ that consists of 
the rows of $\H_k$ indexed by the set $S$, and $\lambda_{max}(\H_k)$
and $\lambda_{min}(\H_k)$ are the maximum and minimum eigenvalues of $\H_k$, respectively. Spectral ($\ell_2$) norm of a matrix is denoted by $\|.\|$. $\I_n \in \R^{n\times n}$ is the identity matrix. Moreover, let $[n] := \{1,2,\dots,n\}$.
\vspace{-0.2cm}
\section{System Model and Problem Formulation}\label{sec:sys}
This section starts by a description of a motivating example of multi-object tracking under communication 
and power constraints. Then we proceed to define the system model and mathematically formulate the sensor 
selection problem studied in the paper.
\vspace{-0.2cm}
\subsection{Motivating example: Accelerated multi-object tracking}\label{sec:uav}
Consider a tracking system, shown in Fig. \oldref{fig:drone}, where a control unit surveys an area via a swarm of unmanned aerial vehicles (UAVs). The UAVs are equipped with GPS and radar systems and can communicate with each other over locally established communication channels. However, only a few of the UAVs known as {\it swarm leaders} are allowed to communicate to the control unit because of various practical restrictions such as power constraints. The UAVs patrol the area according to a predefined search pattern to gather information about the location of mobile objects of interest. Each UAV, by using its radar system, acquires range and angular measurements of all the objects that are within the maximum radar detection range and are capable of transmitting those measurements to the swarm leaders. Due to limitations on the rate of communication between the swarm leaders and the control unit, and to reduce delays in tracking from high computation, only a subset of the gathered measurements is communicated to the control unit. In order to track the locations of the object, the control unit employs Kalman filtering using the received measurements. Therefore, the goal of swarm leaders is to perform {\it sensor scheduling} and select a subset of range and angular measurements such that (i) the communication constraint is satisfied, and (ii) the mean-square error of the Kalman filter estimate of the objects' locations is minimized.
\vspace{-0.2cm}
\subsection{System model}
Consider a discrete-time, linear, time-varying state-space model described by
\begin{equation}\label{eq:sys}
\begin{aligned}
\x_{k+1} &= \mathbf{A}_k\x_k+\w_k \\
\y_k &= \H_k\x_k+\v_k,
\end{aligned}
\end{equation}
where $\x_k \in \R^m$ is the state vector at time $k$ that we aim to estimate, $\y_k \in \R^n$ is the 
measurement vector, $\w_k \in \R^m$ and $\v_k \in \R^n$ are zero-mean white Gaussian noise
processes with covariances $\mathbf{Q}_k$ and $\mathbf{R}_k$, respectively, $\mathbf{A}_k \in  
\R^{m \times m}$ is the state transition matrix and $\H_k \in  \R^{n \times m}$ is the matrix 
whose rows at time $k$ are the measurement vectors $\h_{k,i} \in \R^m$, $1 \le i \le n$. 
We assume the states $\x_k$ are uncorrelated with $\w_k$ and $\v_k$. Additionally, we 
assume that $\x_0 \sim {\cal N}(0, \mathbf{\Sigma}_x)$, $\mathbf{Q}_k=\sigma^2\I_m$, and 
$\mathbf{R}_k=\mathrm{diag}(\sigma_1^2,\dots,\sigma_n^2)$. Note that, unlike the past work
on greedy sensor selection in
\cite{shamaiah2010greedy,shamaiah2012greedy,chamon2017approximate,chamon2017mean},
this model does not restrict the noise covariance matrix to be a multiple of identity.

Due to limited resources, fusion center aims to select $K$ out of $n$ sensors and use their 
measurements to estimate the state vector $\x_k$ such that the trace of the covariance matrix 
of the estimation error, i.e., the MSE of the estimator implemented using the Kalman 
filter is minimized. Similar to prior work in 
\cite{joshi2009sensor,shamaiah2010greedy,shamaiah2012greedy}, we assume that the 
measurement vectors $\h_{k,i}$ are available at the fusion center. Let $\hat{\x}_{k|k-1}$ and
$\hat{\x}_{k|k}$ denote the predicted and filtered linear minimum mean-square error (LMMSE) estimators of 
$\x_k$, respectively. In other words,
$\hat{\x}_{k|k-1}$ is the LMMSE estimator of $\x_k$ given $\{\y_{S_1},\dots,\y_{S_{k-1}} \}$
and
$\hat{\x}_{k|k}$ is the LMMSE estimator of $\x_k$  given  $\{\y_{S_1},\dots,\y_{S_k} \}$,
where $S_j$ denotes the set of sensors selected at time $j$ and $\y_{S_j}$ denotes the vector
of measurements collected by those sensors.
Moreover, let $\P_{k|k-1}$ and $\P_{k|k}$ denote the predicted and filtered error covariance 
matrix of the Kalman filter at time instant $k$, respectively, i.e.,
\begin{equation*}
	\begin{aligned}
		\P_{k|k-1} &= \mathbf{A}_k\P_{k-1|k-1}\mathbf{A}_k^\top+\mathbf{Q}_k, \\
		\P_{k|k} &= \left(\P_{k|k-1}^{-1}+\H_{k,S_{k}}^\top\mathbf{R}_{k,S_{k}}^{-1}\H_{k,S_{k}}\right)^{-1},
	\end{aligned}
\end{equation*}
where 
$P_{0|0}=\mathbf{\Sigma}_x$. Since $\mathbf{R}_k=\mathrm{diag}(\sigma_1^2,\dots,\sigma_n^2)$ 
and the measurements are uncorrelated across sensors, it holds that
\begin{equation*}
	\begin{aligned}
		\P_{k|k} = \left(\P_{k|k-1}^{-1}+\H_{k,S_{k}}^\top\mathrm{diag}(\{\sigma_i^{-2}\}_{i\in S_{k}})\H_{k,S_{k}}\right)^{-1}.
	\end{aligned}
\end{equation*}
Furthermore, $\F_{S_k} = \P_{k|k}^{-1} = \P_{k|k-1}^{-1}+\sum_{i\in S_{k}}\sigma_i^{-2}\h_{k,i}
\h_{k,i}^\top$ is the corresponding Fisher information matrix. In the information form, the filtered 
estimator of $\x_k$ is expressed as
\begin{equation}\label{eq:estimator}
\begin{aligned}
\hat{\x}_{k|k} = 
\F_{S_k}^{-1} \H_{k,S_{k}}^\top\mathrm{diag}(\{\sigma_i^{-2}\}_{i\in S_{k}}) \y_k.
\end{aligned}
\end{equation}
The MSE of the estimate found in \ref{eq:estimator} is given by the trace of the filtered error covariance 
matrix $\P_{k|k}$:
\begin{equation}\label{eq:mse}
\begin{aligned}
\text{MSE}_{S_k} = \E\left[{\norm{\x_k-\hat{\x}_{k|k}}}_2^2\right]= \text{Tr}\left(\F_{S_k}^{-1}\right).
\end{aligned}
\end{equation}
To minimize \ref{eq:mse}, at each time step the fusion center seeks a solution to the optimization 
problem
\begin{equation}\label{eq:sensor1}
\begin{aligned}
& \underset{S}{\text{min}}
\quad \text{Tr}\left(\F_S^{-1}\right)
& \text{s.t.}\hspace{0.5cm}  S \subset [n], \phantom{k}|S|=K.
\end{aligned}
\end{equation}
By a reduction to the well-known set cover problem, the combinatorial optimization \ref{eq:sensor1}
can be shown to be NP-hard \cite{feige1998threshold,williamson2011design}. In principle, to find the 
optimal solution one needs to exhaustively search over all schedules of $K$ sensors. The 
techniques proposed in \cite{joshi2009sensor}, albeit for an optimality criterion different from MSE 
and a simpler measurement model, find a subset of sensors that yields a sub-optimal MSE 
performance while being computationally much more efficient than the exhaustive search. In particular,
\cite{joshi2009sensor} relies on finding the solution to the following SDP relaxation:
\begin{equation}\label{sdp}
\begin{aligned}
& \underset{\z_k,\Y}{\text{min}}
\quad \text{Tr}(\Y)\\
& \text{s.t.}\hspace{0.5cm}  0\leq z_{k,i} \leq 1, \phantom{k} \forall i\in [n]\\
& \hspace{0.9cm} \sum_{i=1}^n z_{k,i} =K\\
& \hspace{0.9cm} 
\begin{bmatrix}
\Y& \I\\
\I & \P_{k|k-1}^{-1}+\sum_{i=1}^n z_{k,i}\sigma_i^{-2}\h_{k,i}\h_{k,i}^\top
\end{bmatrix} \succeq \mathbf{0}.
\end{aligned}
\end{equation}
The complexity of the SDP algorithm scales as ${\cal O}(n^3)$ which is infeasible in many practical
settings. Furthermore, there are no guarantees on the achievable MSE performance of the SDP 
relaxation. Note that when the number of sensors in a network and the size of the state vector $\x_k$ 
are relatively large, even the greedy algorithm proposed in \cite{shamaiah2010greedy} may be 
computationally prohibitive. 
\vspace{-0.2cm}
\section{Sensor Selection via Optimizing a Weak Submodular Objective}\label{sec:alg}
Leveraging the idea of {\it weak submodularity}, in this section we propose a new 
formulation of the sensor selection problem concerned with minimizing the 
MSE of the Kalman filter that relies on a subset of network nodes to track states
of a hidden random process. We first overview concepts that are essential for the 
development of the proposed framework.
\begin{definition}
	\label{def:unif}
	A set function $f:2^X\rightarrow \mathbb{R}$ is monotone non-decreasing if $f(S)\leq f(T)$ for all $S\subseteq T\subseteq X$.
\end{definition}
\begin{definition}
	\label{def:submod}
	A set function $f:2^X\rightarrow \mathbb{R}$ is submodular if 
	\begin{equation}
	f(S\cup \{j\})-f(S) \geq f(T\cup \{j\})-f(T)
	\end{equation}
	for all subsets $S\subseteq T\subset X$ and $j\in X\backslash T$.
\end{definition}

A concept closely related to submodularity is the notion of curvature of a set function. 
The curvature quantifies how close the function is to being submodular. In particular, here
we define the element-wise curvature.
\begin{definition}
The element-wise curvature ${\cal C}_l$ of a monotone non-decreasing function $f$ is defined as
\begin{equation}
{\cal C}_l=\max_{(S,T,i)\in \mathcal{X}_l}{f_i(T)\slash f_i(S)},
\end{equation}
where $f_i(S)=f(S\cup \{i\})-f(S)$ and $f_i(T)=f(T\cup \{i\})-f(T)$ denote the marginal 
values of adding element $i$ to sets $S$ and $T$, respectively, and
$\mathcal{X}_l = \{(S,T,i)|S \subset T \subset X, i\in X \backslash T, |T 
\backslash S|=l,|X|=n\}$. Furthermore, the maximum element-wise curvature is 
denoted by ${\cal C}_{\max}=\max_{l} {{\cal C}_l}$.
\end{definition}
A set function is submodular if and only if ${\cal C}_{\max} \le 1$. We refer to $f(S)$ 
as being weak submodular if its curvature ${\cal C}_{\max} > 1$ is bounded above.

\begin{definition}
Let $X$ be a finite set and let $\mathcal{I}$ denote a collection of subsets of $X$. 
The pair $\mathcal{M}=(X,\mathcal{I})$ is a matroid if the following two statements hold: 
\begin{itemize}
\item Hereditary property. If $T\in \mathcal{I}$, then $S\in \mathcal{I}$ for all $S\subseteq T$.
\item Augmentation property. If $S,T\in \mathcal{I}$ and $\abs{S}<\abs{T}$, then there exists $e\in T\backslash S$ such that $S\cup \{e\}\in \mathcal{I}$.
\end{itemize}
The collection $\mathcal{I}$ is called the set of independent sets of the matroid $\mathcal{M}$. A maximal independent set is a basis. It is easy to show that all the bases of a matroid have the same cardinality.
\end{definition}

Given a monotone non-decreasing set function $f:2^X\rightarrow \mathbb{R}$ with $f(\emptyset)=0$, and a uniform matroid $\mathcal{M}=(X,\mathcal{I})$, we
are interested in solving the combinatorial problem
\begin{equation}\label{eq:fmax}
\max_{S \in \mathcal{I}} f(S).
\end{equation}


Recall that for Kalman filtering in the resource constrained scenario, if $S_k$ is the set of 
sensors selected at time $k$ then the error covariance matrix of the filtered estimate is 
$\P_{k|k} = \F_{S_k}^{-1}$, the inverse of the corresponding Fisher information matrix. Let us 
define $f(S)$ as
\[
f(S) = \mathrm{Tr}\left(\P_{k|k-1}-\F_S^{-1}\right). 
\]
Clearly, since $\P_{k|k-1}$ is known, there is a one-to-one correspondence between $f(S_k)$ 
computed for a given subset of sensors $S_k$ and the MSE of the LMMSE estimator
(i.e., filtered estimate of the Kalman filter) that uses measurements acquired by the sensors 
in $S_k$. Therefore, we can express the optimization problem \ref{eq:sensor1} as
\begin{equation}\label{eq:sensor}
\begin{aligned}
& \underset{S}{\text{max}}
\quad f(S)
& \text{s.t.}\hspace{0.5cm}  S \subset [n], \phantom{k}|S|=K.
\end{aligned}
\end{equation}
We now argue that \ref{eq:sensor} is indeed an instance of the general combinatorial problem \ref{eq:fmax}.
By defining $X = [n]$ and $\mathcal{I} = \{S \subset X| |S|=K\}$, it is easy to see that $\mathcal{M}=(X,\mathcal{I})$ is a matroid.  In Proposition \oldref{p:mono} below we characterize important properties of $f(S)$ and develop a recursive scheme to efficiently compute the marginal gain of querying a sensor. The formula for the marginal gain of $f(S)$ is also of interest in our subsequent analysis of its weak submodularity properties.
\begin{proposition}\label{p:mono}
\textit{Let $f(S) = \mathrm{Tr}\left(\P_{k|k-1}-\F_S^{-1}\right)$. Then, $f(S)$ is a monotonically increasing set function, $f(\emptyset)=0$, and 
\begin{equation}\label{eq:mg}
f_j(S) = \frac{\h_{k,j}^\top\F_S^{-2}\h_{k,j}}{\sigma_j^{2}+\h_{k,j}^\top\F_S^{-1}\h_{k,j}},
\end{equation}
where upon adding element $j$ to $S$, $\F_{S}$ is updated according to
\begin{equation}\label{eq:upf}
\F_{S \cup\{j\}}^{-1} = \F_{S}^{-1}-\frac{\F_{S}^{-1}\h_{k,j}\h_{k,j}^\top\F_{S}^{-1}}{\sigma_j^2+\h_{k,j}^\top\F_{S}^{-1}\h_{k,j}}.
\end{equation}}
\end{proposition}
\begin{proof}
See Appendix~\oldref{pf:mono}.
\end{proof}

As stated in Section~I, the MSE is not supermodular \cite{olshevsky2017non,krause2008near}. Consequently, the proposed objective $f(S) = \mathrm{Tr}\left(\P_{k|k-1}-\F_S^{-1}\right)$ is also not submodular. However, as we show 
in Theorem \oldref{thm:curv}, under certain conditions $f(S)$ is characterized by a bounded maximum element-wise curvature $\mathcal{C}_{\max}$. Theorem \oldref{thm:curv} also 
states a probabilistic theoretical upper bound on $\mathcal{C}_{\max}$ in scenarios where 
at each time step the measurement vectors $\h_{k,j}$'s are realizations of independent identically distributed (i.i.d.) random 
vectors drawn from a suitable distribution.

Before proceeding to Theorem \oldref{thm:curv} and its proof, we first state the matrix 
Bernstein inequality \cite{tropp2015introduction} and Weyl's inequality 
\cite{bellman1997introduction} which we will later use in the proof of Theorem \oldref{thm:curv}.
\begin{lemma}\label{thm:ber}
	(Matrix 
	Bernstein inequality \cite{tropp2015introduction}) Let $\{\X_\ell\}_{\ell=1}^n$ be a finite collection of independent, random, Hermitian matrices in $\R^{m\times m}$. Assume that for all $\ell \in [n]$,
	\begin{equation}\label{eq:ber}
	\E\left[\X_\ell\right] = \mathbf{0}, \quad \lambda_{max}(\X_\ell)\leq L.
	\end{equation}
	Let $\Y=\sum_{\ell=1}^n \X_\ell$. Then, for all $q>0$, it holds that
	\begin{equation}
	\Pr\{\lambda_{max}(\Y)\geq q \} \leq m\exp\left(\frac{-q^2\slash 2}{\|\E\left[\Y^2\right]\|+Lq\slash 3}\right).
	\end{equation} 
\end{lemma}
\begin{lemma}\label{lem:eig}
	(Weyl's inequality \cite{bellman1997introduction}) Let $\mathbf{A}$ and $\B$ be two $m\times m$ real positive definite matrices. Then it holds that
	\begin{equation}
	\lambda_{l}(\mathbf{A})+\lambda_{\min}(\B) \leq \lambda_{l}(\mathbf{A}+\B) \leq \lambda_{l}(\mathbf{A})+\lambda_{\max}(\B)
	\end{equation}
	where $ \lambda_{l}(\mathbf{A})  $ denotes the $l\ts{th}$ largest eigenvalue of $\mathbf{A}$.
\end{lemma}
We now proceed to the statement and proof of Theorem \oldref{thm:curv}.
\begin{theorem}\label{thm:curv}
Let $\mathcal{C}_{max}$ be the maximum element-wise curvature of $f(S)$, the objective 
function of the sensor selection problem. Assume that $\|\h_{k,j}\|_2^2\leq C$ for all $j$ 
and $k$. If 
\begin{equation}\label{eq:fstcond}
\lambda_{max}(\H_k^\top\H_k) \leq \left(\frac{1}{\phi} -\frac{1}{\lambda_{min}(\P_{k|k-1})}\right)\min_{j\in[n]}\sigma_j^2
\end{equation}
for some $0<\phi< \lambda_{min}(\P_{k|k-1})$, then it holds that 
\begin{equation}\label{eq:phic}
\mathcal{C}_{max} \leq\max_{j\in[n]} \frac{\lambda_{max}(\P_{k|k-1})^{2}(\sigma_j^{2}+
\lambda_{max}(\P_{k|k-1})C)}{\phi^{2}(\sigma_j^{2}+\phi C)}.
\end{equation}
Furthermore, if $\h_{k,j}$'s are i.i.d. 
zero-mean random vectors with covariance matrix $\sigma_h^2\I_m$ such that 
$\sigma_h^2<C$, then for all $q>0$, with probability 
	\begin{equation} \label{eq:psuc}
	p\geq 1-m\exp\left(\frac{-q^2\slash 2}{(C-\sigma_h^2)(n\sigma_h^2+q\slash 3)}\right),
	\end{equation}
    it holds that
	\begin{equation}\label{eq:phi}
	\phi = \min_{j\in[n]} \left(\frac{1}{\lambda_{min}(\P_{k|k-1})}+\frac{n\sigma_h^2+q}{\sigma_j^2}\right)^{-1}>0.
	\end{equation}
\end{theorem}
\begin{proof}
We prove the statement of the theorem by relying on the recursive expression for the marginal 
gain stated in Proposition~1. We first establish a sufficient condition for weak submodularity of 
$f(S)$. In particular, from the definition of the element-wise curvature and \ref{eq:mg}, for all 
$(S,T,j)\in \mathcal{X}_l$ we obtain
\begin{equation}
\begin{aligned}
{\cal C}_l&=\max_{(S,T,j)\in \mathcal{X}_l}{\frac{(\h_{k,j}^\top\F_T^{-2}\h_{k,j})(\sigma_j^{2}+\h_{k,j}^\top\F_S^{-1}\h_{k,j})}{(\h_{k,j}^\top\F_S^{-2}\h_{k,j})(\sigma_j^{2}+\h_{k,j}^\top\F_T^{-1}\h_{k,j})}}\\\vspace{0.2cm}
&\leq \max_{(S,T,j)\in \mathcal{X}_l}{\frac{\lambda_{max}(\F_T^{-2})(\sigma_j^{2}+\lambda_{max}(\F_S^{-1})\|\h_{k,j}\|_2^2)}{\lambda_{min}(\F_S^{-2})(\sigma_j^{2}+\lambda_{min}(\F_T^{-1})\|\h_{k,j}\|_2^2)}},
\end{aligned}
\end{equation}
where the inequality follows from the Courant–Fischer min-max theorem \cite{bellman1997introduction}. Notice that $\lambda_{max}(\F_S^{-1})=\lambda_{min}(\F_S)^{-1}$ and $\lambda_{min}(\F_T)\geq\lambda_{min}(\F_S)\geq\lambda_{min}(\F_\emptyset)=\lambda_{min}(\P_{k|k-1}^{-1})$ by Lemma \oldref{lem:eig}. This fact, along with the definition of $\mathcal{C}_{max}$ implies
	\begin{equation}\label{eq:ray}
	\begin{aligned}
	\mathcal{C}_{max} &\leq\max_{j \in[n]} \frac{\lambda_{max}(\P_{k|k-1})^{2}(\sigma_j^{2}+\lambda_{max}(\P_{k|k-1})\|\h_{k,j}\|_2^2)}{\lambda_{max}(\F_S)^{-2}(\sigma_j^{2}+\lambda_{max}(\F_T)^{-1}\|\h_{k,j}\|_2^2)}\\\vspace{0.2cm}
	&\stackrel{(a)}{\leq}\max_{j \in[n]} \frac{\lambda_{max}(\P_{k|k-1})^{2}(\sigma_j^{2}+\lambda_{max}(\P_{k|k-1})\|\h_{k,j}\|_2^2)}{\lambda_{max}(\F_{[n]})^{-2}(\sigma_j^{2}+\lambda_{max}(\F_{[n]})^{-1}\|\h_{k,j}\|_2^2)}\\\vspace{0.2cm}
	&\stackrel{(b)}{\leq}\max_{j \in[n]} \frac{\lambda_{max}(\P_{k|k-1})^{2}(\sigma_j^{2}+\lambda_{max}(\P_{k|k-1})C)}{\lambda_{max}(\F_{[n]})^{-2}(\sigma_j^{2}+\lambda_{max}(\F_{[n]})^{-1}C)},
	\end{aligned}
	\end{equation}
	where (a) follows from the fact that $\lambda_{max}(\F_S)\leq\lambda_{max}(\F_T)\leq\lambda_{max}(\F_{[n]})$ and (b) holds since
	\begin{equation}
	g(x) = \frac{\sigma_j^{2}+\lambda_{max}(\P_{k|k-1})x}{\sigma_j^{2}+\lambda_{max}(\F_{[n]})^{-1}x}
	\end{equation} 
	is a monotonically increasing function for $x>0$. Now, since the maximum eigenvalue of a positive definite matrix satisfies the triangle inequality, we have
	\begin{equation}
	\begin{aligned}
	\lambda_{max}(\F_{[n]})&\leq \frac{1}{\lambda_{min}(\P_{k|k-1})}+\lambda_{max}(\sum_{j=1}^n \frac{1}{\sigma_j^2}\h_{k,j}\h_{k,j}^\top)\\
	&\leq\frac{1}{\lambda_{min}(\P_{k|k-1})}+\max_{j\in[n]}\frac{1}{\sigma_j^2}\lambda_{max}(\H_k^\top\H_k).
	\end{aligned}
	\end{equation}
Therefore, by combining inequalities \ref{eq:fstcond} and \ref{eq:ray} we obtain the result in
\ref{eq:phic}. 

Next, to analyze the setting of i.i.d random measurement vectors, 
we bound $\lambda_{max}(\F_{[n]})$ using Lemma \oldref{thm:ber}.
	Let $\X_j=\h_{k,j}\h_{k,j}^\top-\sigma_h^2\I_m$ and $\Y=\sum_{j=1}^n \X_j$. To use the result of Lemma \oldref{thm:ber}, one should first verify expressions in \ref{eq:ber}. To this end, note that
	\begin{equation}
	\begin{aligned}
	\E[\X_j]&= \E[\h_{k,j}\h_{k,j}^\top-\sigma_h^2\I_m] \\
	&=\E[\h_{k,j}\h_{k,j}^\top] -\sigma_h^2\I_m =\mathbf{0}. 
	\end{aligned}
	\end{equation}
	This  in turn  implies that $\E[\Y]=\mathbf{0}$.
	Since $\X_j$'s are independent,
	\begin{equation}
	\|\E[\Y^2]\|= \|\E[\sum_{j=1}^n \X_j^2]\|\leq \sum_{j=1}^n \|\E[\X_j^2]\|
	\end{equation}
	by the linearity of expectation and the triangle inequality. To proceed, we need to 
	determine $\lambda_{max}(\X_j)$ and $\E[\X_j^2]$. First, let us verify $\h_{k,j}$ is an 
	eigenvector of $\X_j$ by observing that
	\begin{equation}
	\begin{aligned}
	\X_j\h_j &= \left(\h_{k,j}\h_{k,j}^\top-\sigma_h^2\I_m\right)\h_{k,j}\\
	& = \left(\|\h_{k,j}\|_2^2-\sigma_h^2\right)\h_{k,j},
	\end{aligned}
	\end{equation} 
	where $\h_{k,j}\h_{k,j}^\top-\sigma_h^2\I_m$ is the corresponding eigenvalue. 
	Since $\h_{k,j}\h_{k,j}^\top$ is a rank-1 matrix, other eigenvalues of $\X_j$ are all equal 
	to $-\sigma_h^2$. Hence, 
	\begin{equation}
	\lambda_{max}(\X_j) \leq C-\sigma_h^2, 
	\end{equation}
	and we recall that $C-\sigma_h^2 > 0$. We can now establish an upper bound on $\E[\X_j^2]$ as
	\begin{equation}
	\begin{aligned}
	\E[\X_j^2] &= \E[\left(\h_{k,j}\h_{k,j}^\top-\sigma_h^2\I_m\right)\left(\h_{k,j}\h_{k,j}^\top-\sigma_h^2\I_m\right)]\\
	& = \left(\|\h_{k,j}\|_2^2-\sigma_h^2\right)\E[ \h_{k,j}\h_{k,j}^\top] \\
	& \quad - \sigma_h^2\E[ \left(\h_{k,j}\h_{k,j}^\top-\sigma_h^2\I_m\right)]\\
	& = \left(\|\h_{k,j}\|_2^2-\sigma_h^2\right)\sigma_h^2\I_m\preceq (C-\sigma_h^2)\sigma_h^2\I_m,
	\end{aligned}
	\end{equation}
	where we have used the fact that $\E[\X_j]=\mathbf{0}$. Thus, $L =C-\sigma_h^2$ and $\|\E[\Y^2]\| \leq n(C-\sigma_h^2)\sigma_h^2$. Now, according to Lemma \oldref{thm:ber}, for all $q>0$ it holds that $\Pr\{\lambda_{max}(\Y)\leq q\}\geq p$ where 
	\begin{equation}\label{eq:prob}
	p= 1-m\exp\left(\frac{-q^2\slash 2}{(C-\sigma_h^2)(n\sigma_h^2+q\slash 3)}\right).
	\end{equation}
	Therefore, 
	\begin{equation}
	\lambda_{max}(\F_{[n]})\leq \frac{1}{\lambda_{min}(\P_{k|k-1})}+\max_{j\in[n]}\frac{n\sigma_h^2+q}{\sigma_j^2}= \phi^{-1}
	\end{equation}
	with probability $p$. This completes the proof.
\end{proof}
\textit{Remark 1:} The setting of i.i.d. random vectors described in Theorem \oldref{thm:curv} arises 
in scenarios where sketching techniques, such as random projections, are used to reduce dimensionality 
of the measurement equation (see \cite{berberidis2016data} for more details). The following are often encountered examples of such settings:
\begin{enumerate}
\item {\it Multivariate Gaussian measurement vectors:} Let $\h_{k,j} \sim \mathcal{N}(0,\frac{1}{m}\I_m)$ 
for all $j$. It is straightforward to show that $\E[\|\h_{k,j}\|_2^2] = 1$. Furthermore, it can be shown that 
$\|\h_{k,j}\|_2^2$ is with high probability concentrated around its expected value. Therefore, for this case,
$\sigma_h^2 = \frac{1}{m}$ and $C = 1$.
\item {\it Centered Bernoulli measurement vectors:} Let each entry of $\h_{k,j}$ be $\pm \frac{1}{\sqrt{m}}$ 
with equal probability. Therefore, $\|\h_{k,j}\|_2^2=1 = C $. Additionally, $\sigma_h^2 = \frac{1}{m}$ since 
the entries of $\h_{k,j}$ are i.i.d. zero-mean random variables with variance $\frac{1}{m}$.
\end{enumerate}

We can interpret the conditions stated in Theorem \oldref{thm:curv} as requirements on the 
condition number of  $\P_{k|k-1}$ as argued next. For a sufficiently large $m$ and 
$\sigma_h^2 = \frac{1}{m}$, it holds that $C \approx 1$. Assume $\phi \geq \lambda_{max}(\P_{k|k-1})\slash \Delta$ for some $\Delta>1$, and $\sigma_j^2 = \sigma^2$ for all $i\in[n]$. Define
\begin{equation}
\mathrm{SNR} = \frac{\lambda_{max}(\P_{k|k-1})}{\sigma^2},  
\end{equation}
and let
\begin{equation}
\kappa = \frac{\lambda_{max}(\P_{k|k-1})}{\lambda_{min}(\P_{k|k-1})} \geq 1 
\end{equation}
be the condition number of $\P_{k|k-1}$. Then, following some elementary numerical approximations, 
we obtain the following corollary.
\begin{corollary}\label{col:con}
Let	
\begin{equation}
\Delta \geq \kappa +c_1\frac{n}{m}\mathrm{SNR}
\end{equation}
for some $c_1>1$. Then with probability 
\begin{equation}
p\geq 1-m\exp(-\frac{n}{m}c_2)
\end{equation}
it holds that $\mathcal{C}_{max} \leq \Delta^3$ for some $c_2>0$.
\end{corollary}
Informally, Theorem \oldref{thm:curv} states that for a well-conditioned $\P_{k|k-1}$ the curvature of 
$f(S)$ is small, which implies weak submodularity of $f(S)$. Furthermore, the probability of such an
event exponentially increases with the number of available measurements. 

\vspace{-0.2cm}
\section{Randomized Greedy Sensor Selection}\label{sec:anl}
\renewcommand\algorithmicdo{}	
\begin{algorithm}[t]
	\caption{Randomized Greedy Sensor Scheduling}
	\label{alg:greedy}
	\begin{algorithmic}[1]
		\STATE \textbf{Input:}  $\P_{k|k-1}$, $\H_k$, $K$, $\epsilon$.
		\STATE \textbf{Output:} Subset $S_k\subseteq [n]$ with $|S_k|=K$.
		\STATE Initialize $S_k^{(0)} =  \emptyset$, $\F_{S_k^{(0)}}^{-1}=\P_{k|k-1}$.
		\FOR{$i = 0,\dots, K-1$}
		\STATE Choose $R$ by sampling $s=\frac{n}{K}\log{(1/\epsilon)}$ indices uniformly at random from $[n]\backslash S_k^{(i)}$.\vspace{0.2cm}
		\STATE $i_s = \argmax_{j\in R} \frac{\h_{k,j}^\top\F_{S_k^{(i)}}^{-2}\h_{k,j}}{\sigma_j^{2}+\h_{k,j}^\top\F_{S_k^{(i)}}^{-1}\h_{k,j}}$.\vspace{0.2cm}
		\STATE Set $S_k^{(i+1)}= S_k^{(i)}\cup \{i_s\}$.\vspace{0.2cm}
		\STATE $\F_{S_k^{(i+1)}}^{-1} = \F_{S_k^{(i)}}^{-1}-\frac{\F_{S_k^{(i)}}^{-1}\h_{k,i_s}\h_{k,i_s}^\top\F_{S_k^{(i)}}^{-1}}{\sigma_j^2+\h_{k,i_s}^\top\F_{S_k^{(i)}}^{-1}\h_{k,i_s}}$
		\ENDFOR
		\RETURN $S_k = S_k^{(K)}$.
	\end{algorithmic}
\end{algorithm}

The complexity of SDP relaxation and greedy algorithms for sensor selection become
prohibitive in large-scale systems. Motivated by the need for practically feasible 
schemes, we present a randomized greedy algorithm for finding an approximate solution 
to \ref{eq:sensor} and derive its performance guarantees. In particular, inspired by the 
technique in \cite{mirzasoleiman2014lazier} proposed in the context of optimizing 
submodular objective functions, we develop a computationally efficient randomized 
greedy algorithm (see Algorithm \oldref{alg:greedy}) that finds an approximate solution 
to \ref{eq:sensor} with a guarantee on the achievable MSE performance of the Kalman 
filter that uses only the observations of the selected sensors. Algorithm \oldref{alg:greedy} 
performs the task of sensor scheduling in the following way. At each iteration of the 
algorithm, a subset $R$ of size $s$ is sampled uniformly at random and without 
replacement from the set of available sensors. The marginal gain provided by each of 
these $s$ sensors to the objective function is computed using \ref{eq:mg}, and the one 
yielding the highest marginal gain is added to the set of selected sensors. Then the 
efficient recursive formula in \ref{eq:upf} is used to update $\F_{S}^{-1}$ so it can 
be analyzed when making the selection in the next iteration. This procedure is 
repeated $K$ times.

\textit{Remark 2:} The parameter $\epsilon$ in Algorithm \oldref{alg:greedy}, $e^{-K}
\leq\epsilon<1$, is a predefined constant that is chosen to strike a desired balance 
between performance and complexity. When $\epsilon = e^{-K}$, each iteration includes 
all of the non-selected sensors in $R$ and Algorithm \oldref{alg:greedy} coincides with 
the conventional greedy scheme. However, as $\epsilon$ approaches $1$, $|R|$ and 
thus the overall computational complexity decreases. 
\vspace{-0.2cm}
\subsection{Performance analysis of the proposed scheme}
In this section we analyze Algorithm \oldref{alg:greedy} and in Theorem \oldref{thm:card} 
provide a bound on the performance of the proposed randomized greedy scheme when 
applied to finding an approximate solution to maximization problem \ref{eq:sensor}.

Before deriving the main result, we first provide two lemmas. Lemma \oldref{lem:curv} 
states an upper bound on the difference between the values of the objective function 
corresponding to two sets having different cardinalities while Lemma \oldref{lem:rand} 
provides a lower bound on the expected marginal gain.

\begin{lemma}\label{lem:curv}
\textit{Let $\{{\cal C}_l\}_{l=1}^{n-1}$ denote the element-wise curvatures of $f(S)$. Let 
$S$ and $T$ be any subsets of sensors such that $S\subset T \subseteq [n]$ with 
$|T\backslash S|=r$. Then it holds that
\begin{equation}
f(T)-f(S)\leq  C(r)\sum_{j\in T\backslash S}f_j(S),
\end{equation}
where $ C(r)=\frac{1}{r}(1+\sum_{l=1}^{r-1}{\cal C}_l)$.}
\end{lemma}
\begin{proof}
See Appendix~\oldref{pf:curv}.
\end{proof}
\begin{lemma}\label{lem:rand}
\textit{Let $S_k^{(i)}$ be the set of selected sensors at the end of the $i\ts{th}$ iteration of Algorithm \oldref{alg:greedy}. Then
\begin{equation}
\E\left[f_{(i+1)_s}(S_k^{(i)})|S_k^{(i)}\right]\geq \frac{1-\epsilon^{\beta}}{K}\sum_{j\in O_k\backslash S_k^{(i)}}f_j(S_k^{(i)}),
\end{equation}
where $O_k$ is the set of optimal sensors at time $k$, $(i+1)_s$ is the index of the selected 
sensor at the $(i+1)\ts{st}$ iteration, $\beta=1+\max\{0,\frac{s}{2n}-\frac{1}{2(n-s)}\}$, and 
$s=\frac{n}{K}\log{(1/\epsilon)}$.}
\end{lemma}
\begin{proof}
See Appendix~\oldref{pf:rand}.
\end{proof}
Theorem \oldref{thm:card} below specifies how accurate the approximate solution to the 
sensor selection problem found by Algorithm \oldref{alg:greedy} is. In particular, if $f(S)$ is 
characterized by a bounded maximum element-wise curvature, Algorithm \oldref{alg:greedy} 
returns a subset of sensors yielding an objective that is on average within a multiplicative 
factor of the objective achieved by the optimal schedule.

\begin{theorem}\label{thm:card}
\textit{Let $\mathcal{C}_{max}$ be the maximum element-wise curvature of $f(S)$, i.e., the 
objective function of sensor scheduling problem in \ref{eq:sensor}. Let $S_k$ denote the 
subset of sensors selected by Algorithm \oldref{alg:greedy} at time $k$, and let $O_k$ be 
the optimum solution to \ref{eq:sensor} such that $|O_k|=K$. Then $f(S_k)$ is on expectation 
a multiplicative factor away from $f(O_k)$. That is,}
\begin{equation}\label{eq:card}
\E\left[f(S_k)\right]\geq \left(1-e^{-\frac{1}{c}}-\frac{\epsilon^\beta}{c}\right) f(O_k),
\end{equation}
\textit{where $c=\max\{{\cal C}_{\max},1\}$, $e^{-K}\leq\epsilon<1$, and 
$\beta=1+\max\{0,\frac{s}{2n}-\frac{1}{2(n-s)}\}$. Furthermore, the computational complexity 
of Algorithm \oldref{alg:greedy} is ${\cal O}(nm^2\log(\frac{1}{\epsilon}))$ where $n$ is the 
total number of sensors and $m$ is the dimension of $\x_k$.}
\end{theorem}
\begin{proof}
Consider $S_k^{(i)}$, the set generated by the end of the $i\ts{th}$ iteration of 
Algorithm \oldref{alg:greedy}. Employing Lemma \oldref{lem:curv} with $S=S_k^{(i)}$ 
and $T=O_k\cup S_k^{(i)}$, and using monotonicity of $f$, yields
\begin{equation}
\begin{aligned}
\frac{f(O_k)-f(S_k^{(i)})}{\frac{1}{r}\left(1+\sum_{l=1}^{r-1}{\cal C}_l\right)}&\leq \frac{f(O_k\cup S_k^{(i)})-f(S_k^{(i)})}{\frac{1}{r}\left(1+\sum_{l=1}^{r-1}{\cal C}_l\right)}\\
&\leq  \sum_{j\in O_k\backslash S_k^{(i)}}f_j(S_k^{(i)}),
\end{aligned}
\end{equation}
where $|O_k\backslash S_k^{(i)}|=r$. Now, using Lemma \oldref{lem:rand} we obtain
\begin{equation}
\E\left[f_{(i+1)_s}(S_k^{(i)})|S_k^{(i)}\right]\geq \left(1-\epsilon^{\beta}\right)\frac{f(O_k)-f(S_k^{(i)})}{\frac{K}{r}\left(1+\sum_{l=1}^{r-1}{\cal C}_l\right)}.
\end{equation}
Applying the law of total expectation yields
\begin{equation}
\begin{aligned}
\E\left[f_{(i+1)_s}(S_k^{(i)})\right]&=\E\left[f(S_k^{(i+1)})-f(S_k^{(i)})\right]\\
&\geq \left(1-\epsilon^{\beta}\right)\frac{f(O_k)-\E\left[f(S_k^{(i)})\right]}{\frac{K}{r}\left(1+\sum_{l=1}^{r-1}{\cal C}_l\right)}.
\end{aligned}
\end{equation}
Using the definition of the maximum element-wise curvature, we obtain 
\begin{equation}
\frac{1}{r}\left(1+\sum_{l=1}^{r-1}{\cal C}_l\right)\leq \frac{1}{r}(1+(r-1){\cal C}_{\max}) =g(r).
\end{equation}
It is easy to verify, e.g., by taking the derivative, that $g(r)$ is decreasing (increasing) with respect to $r$ if ${\cal C}_{\max}<1$ (${\cal C}_{\max}>1$). Let $c=\max\{{\cal C}_{\max},1\}$. Then 
\begin{equation}
\frac{1}{r} \left(1+\sum_{l=1}^{r-1}{\cal C}_l\right)\leq \frac{1}{r}(1+(r-1){\cal C}_{\max}) \leq c.
\end{equation}
Hence, 
\begin{equation}
\E\left[f(S_k^{(i+1)})-f(S_k^{(i)})\right]\geq \frac{1-\epsilon^\beta}{Kc}\left(f(O_k)-\E\left[f(S_k^{(i)})\right]\right).
\end{equation}
Using an inductive argument and due to the fact that $f(\emptyset) = 0$, we obtain
\begin{equation}
\E[f(S_k)]\geq \left(1-\left(1-\frac{1-\epsilon^\beta}{Kc}\right)^K\right)f(O_k).
\end{equation}
Finally, using the fact that $(1+x)^y\leq e^{xy}$ for $y>0$ and  the easily verifiable fact that $e^{ax}\leq 1+axe^a$ for $0<x<1$, 
\begin{equation}
\begin{aligned}
\E[f(S_k)]&\geq \left(1-e^{-\frac{1-\epsilon^\beta}{c}}\right)f(O_k)\\
&\stackrel{}{\geq} \left(1-e^{-\frac{1}{c}}-\frac{\epsilon^\beta}{c}\right)f(O_k).
\end{aligned}
\end{equation}
To take a closer look at computational complexity of Algorithm \oldref{alg:greedy}, note that 
step 6 costs $\mathcal{O}(\frac{n}{K}m^2\log(\frac{1}{\epsilon}))$ since one needs to compute 
$\frac{n}{K}\log(\frac{1}{\epsilon})$ marginal gains, each requiring ${\cal O}(m^2)$ operations. 
Furthermore, step 8 requires ${\cal O}(m^2)$ arithmetic operations. Since there are $K$ such 
iterations, running time of Algorithm \oldref{alg:greedy} is ${\cal O}(nm^2\log(\frac{1}{\epsilon}))$. 
This completes the proof.
\end{proof}
Using the definition of $f(S)$ we obtain Corollary \oldref{col:mse} stating that, at each time step, the achievable MSE in \ref{eq:mse} obtained by forming an estimate using sensors selected by the randomized greedy algorithm is within a factor of the optimal MSE.

\begin{corollary}\label{col:mse}
Consider the notation and assumptions of Theorem \oldref{thm:card} and introduce
$\alpha =1-e^{-\frac{1}{c}}-\frac{\epsilon^\beta}{c}$. Let $\mathrm{MSE}_{S_k}$ denote the 
mean-square estimation error obtained by forming an estimate using information provided 
by the sensors selected by Algorithm \oldref{alg:greedy} at time $k$, and let 
$\mathrm{MSE}_{o}$ be the optimal mean-square error formed using information collected 
by the sensors specified by the optimum solution of \ref{eq:sensor}. Then the expected 
$\mathrm{MSE}_{S_k}$ is bounded as
\begin{equation}\label{eq:card1}
\E\left[\mathrm{MSE}_{S_k}\right]\leq \alpha \mathrm{MSE}_{o} + (1-\alpha) \mathrm{Tr}(\P_{k|k-1}).
\end{equation}
\end{corollary}

\textit{Remark 3:} Since the proposed sensor selection scheme is a randomized algorithm, 
the analysis of its {\it expected} MSE, as provided by Theorem \oldref{thm:card} and Corollary 
\oldref{col:mse}, is a meaningful performance characterization. Notice that, as expected, 
$\alpha$ is decreasing in both $c$ and $\epsilon$. If $f(S)$ is characterized by a small 
curvature, then $f(S)$ is nearly submodular and the randomized greedy algorithm delivers a 
near-optimal sensor scheduling. As we decrease $\epsilon$, $\alpha$ increases which in turn 
leads to a better approximation factor. Moreover, by following an argument similar to that of 
the classical analysis in \cite{nemhauser1978analysis}, one can show that the approximation 
factor for the greedy algorithm is given by $\alpha_g = 1-e^{-\frac{1}{c}}$ (see also 
\cite{chamon2017mean,singh2017supermodular}). Therefore, the term $\frac{\epsilon}{c}$ 
in $\alpha$ denotes the difference between the approximation factors of the proposed 
randomized greedy algorithm and the conventional greedy scheme.

\textit{Remark 4:} The computational complexity of the greedy method for sensor selection 
that finds marginal gains via the efficient recursion given in Proposition~1 is ${\cal O}(Knm^2)$. 
Hence, our proposed scheme provides a reduction in complexity by 
$K\slash \log(\frac{1}{\epsilon})$ which may be particularly beneficial in large-scale networks, 
as illustrated in our simulation results.

\textit{Remark 5:} In contrast to the results of \cite{mirzasoleiman2014lazier} derived in the
context of maximizing monotone submodular functions, Theorem~\oldref{thm:card} relaxes 
the submodularity assumption and states that the randomized greedy algorithm does not 
require submodularity to achieve near-optimal performance. Rather, if the set function is 
{\it weak submodular}, Algorithm \oldref{alg:greedy} still selects a subset of sensors that 
provide an MSE near that achieved by the optimal subset of sensors. In addition, even if 
the function is submodular (e.g., if we use the $\log\det$ objective instead of the MSE), the 
results of Theorem \ref{thm:card} offer an improvement over the theoretical results of 
\cite{mirzasoleiman2014lazier} due to a tighter approximation bound stemming from the
analysis presented in the proof of Theorem \ref{thm:card}. Moreover, a major assumption 
in \cite{mirzasoleiman2014lazier} is that ${R}$ is constructed by sampling with replacement. 
Clearly, this contradicts the fact that a sensor selected in one iteration will not be in ${R}$ 
in the subsequent iteration with probability one. On the contrary, we assume ${R}$ is 
constructed by sampling without replacement and carry out the analysis in this setting
that matches the actual randomized greedy sensor selection strategy.


The randomized selection step of Algorithm \oldref{alg:greedy} can be interpreted as 
an approximation of the marginal gains of the selected sensors using a  greedy 
scheme \cite{shamaiah2010greedy}. More specifically, for the $i^{th}$ iteration it holds 
that $f_{j_{rg}}(S_k^{(i)}) = \eta_k^{(i)} f_{j_{g}}(S_k^{(i)})$, where subscripts $rg$ and 
$g$ refer to the sensors selected by the randomized greedy (Algorithm \oldref{alg:greedy}) 
and the greedy algorithm, respectively, and $\{\eta_k^{(i)}\}_{i=1}^K$ are random variables with mean $\mu_i(\epsilon)$ that satisfy $0<\ell_{i}(\epsilon)\leq\eta_k^{(i)}\leq 1$ for 
all $i \in [K]$.\footnote{Notice that 
$\ell_{i}(\epsilon)$ and $\mu_{i}(\epsilon)$ are time-varying quantities where the time 
index is omitted for the simplicity of notation.}
In view of this argument, we obtain 
Theorem \oldref{thm:pac} which states that  if $f(S)$ is characterized by a bounded maximum element-wise curvature and $\{\eta_k^{(i)}\}_{i=1}^K$ are independent random variables, Algorithm \oldref{alg:greedy} returns a subset of sensors yielding an objective that with high probability is only a multiplicative factor away from the objective achieved by the optimal schedule.
\begin{theorem}\label{thm:pac}
Instate the notation and assumptions of Theorem \oldref{thm:card}. Let $\{\eta_k^{(i)}\}_{i=1}^K$ 
denote a collection of random variables such that $0<\ell_{i}(\epsilon)\leq\eta_k^{(i)}\leq 1$, and 
$\E[\eta_k^{(i)}] = \mu_i(\epsilon)$ for all $i$ and $k$. Let $\ell_{min}(\epsilon) = 
\min_{i,k}\{\ell_i(\epsilon)\}$ and $\mu_{min}(\epsilon) = \min_{i,k}\{\mu_i(\epsilon)\}$. Then,
\begin{equation}\label{eq:res1}
\begin{aligned}
f(S_k) \geq \left(1- e^{-\frac{\ell_{min}(\epsilon)}{c}}\right)f(O_k).
\end{aligned}
\end{equation}
Furthermore, if  $\{\eta_k^{(i)}\}_{i=1}^K$ are independent, then for all $0<q<1$ with probability 
at least $1-e^{-CK}$, it holds that
\begin{equation}\label{eq:pacbound}
f(S_k)\geq \left(1-e^{-\frac{(1-q)\mu_{min}(\epsilon)}{c}}\right) f(O_k),
\end{equation}
for some $C>0$.
\end{theorem}
\begin{proof}
Consider $S_k^{(i)}$, the set generated by the end of the $i\ts{th}$ iteration of Algorithm 
\oldref{alg:greedy} and let $(i+1)_{g}$ and $(i+1)_{rg}$ denote the sensors selected by the
greedy and randomized greedy algorithm in the $i\ts{th}$ iteration, respectively.  Let 
$c=\max\{{\cal C}_{\max},1\}$. Employing Lemma \oldref{lem:curv} with $S=S_k^{(i)}$ and 
$T=O_k\cup S_k^{(i)}$, and using monotonicity of $f$, yields
\begin{equation}
\begin{aligned}
f(O_k)-f(S_k^{(i)})&\leq f(O_k\cup S_k^{(i)})-f(S_k^{(i)})\\
&\leq  c \sum_{j\in O_k\backslash S_k^{(i)}}f_j(S_k^{(i)}).
\end{aligned}
\end{equation}
Using the fact that 
\begin{equation}
f_j(S_k^{(i)})\leq f_{(i+1)_{rg}}(S_k^{(i)}) \leq f_{(i+1)_{g}}(S_k^{(i)})
\end{equation}
for all $j$, we obtain
\begin{equation}\label{eq:rel1}
\begin{aligned}
f(O_k)-f(S_k^{(i)})\leq cKf_{(i+1)_{g}}(S_k^{(i)}).
\end{aligned}
\end{equation}
On the other hand,
\begin{equation}\label{eq:rel2}
\begin{aligned}
f(S_k^{(i+1)}) - f(S_k^{(i)})&=f_{(i+1)_{rg}}(S_k^{(i)}) \\
&=\eta_k^{(i+1)} f_{(i+1)_{g}}(S_k^{(i)}).
\end{aligned}
\end{equation} 
Combining \ref{eq:rel1} and \ref{eq:rel2} yields
\begin{equation}
f(S_k^{(i+1)}) - f(S_k^{(i)}) \geq \frac{\eta_k^{(i+1)}}{Kc}\left(f(O_k)-f(S_k^{(i)})\right).
\end{equation}
Using an inductive argument similar to the one in the proof of Theorem \oldref{thm:card}, 
and noting that $f(\emptyset) = 0$,
\begin{equation}\label{eq:pacb1}
\begin{aligned}
f(S_k) &\geq  \left(1-\left(1- \sum_{i=1}^K\frac{\eta_k^{(i)}}{Kc}\right)\right)f(O_k)\\
&\stackrel{(a)}{\geq} \left(1- e^{-\sum_{i=1}^K\frac{\eta_k^{(i)}}{Kc}}\right)f(O_k),
\end{aligned}
\end{equation}
where to obtain $(a)$ we use the fact that $(1+x)^y\leq e^{xy}$ for $y>0$. 
Therefore, since by assumption 
$\ell_{min}(\epsilon)\leq\ell_{i}(\epsilon)\leq\eta_k^{(i)}\leq 1$, we establish \ref{eq:res1}.

To show the second statement, i.e., prove \ref{eq:pacbound} holds in the setting of 
independent $\{\eta_k^{(i)}\}_{i=1}^K$,
we apply the Bernstein's inequality \cite{hogg1995introduction} to the sum of independent
random variables $\sum_{i=1}^K\eta_k^{(i)}$. Since $\{\eta_k^{(i)}\}$ are bounded random 
variables, from Popoviciu's inequality \cite{hogg1995introduction} for all $i\in[K]$, it follows
that
\begin{equation}
\Var[\eta_k^{(i)}] \leq \frac{1}{4}(1-\ell_{i}(\epsilon))^2.
\end{equation}
Hence, based on the Bernstein's inequality, for all $0<q<1$
\begin{equation}
\Pr\{\sum_{i=1}^K\eta_k^{(i)}< (1-q)\sum_{i=1}^K\mu_i\} <p,
\end{equation}
where 
\begin{equation}
\begin{aligned}
p&=\exp\left(-\frac{(1-q)^2(\sum_{i=1}^K\mu_i(\epsilon))^2}{\frac{1-q}{3}\sum_{i=1}^K\mu_i(\epsilon)+\frac{1}{4}\sum_{i=1}^K(1-\ell_{i}(\epsilon))^2}\right)\\
&\stackrel{(b)}{\leq} \exp\left(-\frac{K(1-q)^2\mu_{min}^2(\epsilon)}{\frac{1-q}{3}\mu_{min}(\epsilon)+\frac{1}{4}(1-\ell_{min}(\epsilon))^2}\right)\\
&\stackrel{}{=} e^{-C(\epsilon,q)K},
\end{aligned}
\end{equation}
where $(b)$ follows because $p$ increases as we replace $\mu_i(\epsilon)$ and 
$\ell_{i}(\epsilon)$ by their lower bounds. Finally, substituting this results in \ref{eq:pacb1} 
yields 
\begin{equation}
f(S_k) \geq\left(1- e^{-\frac{(1-q)\mu_{min}(\epsilon)}{c}}\right)f(O_k),
\end{equation}
with probability at least $1-e^{C(\epsilon,q)K}$. This completes the proof.
\end{proof}
Our simulation studies in Section \oldref{sec:sim} empirically confirm the results of Theorems 
\oldref{thm:card} and \oldref{thm:pac}, and illustrate that Algorithm \oldref{alg:greedy} performs favorably 
compared to the competing schemes both on average as well as for each individual sensor 
scheduling task.

Similar to Corollary \oldref{col:mse}, we can now obtain a probabilistic bound on the 
MSE \ref{eq:mse} achievable at each time step using the proposed 
randomized greedy algorithm. This result is stated in Corollary \oldref{col:pac} below.

\begin{corollary}\label{col:pac}
Consider the notation and assumptions of Corollary \oldref{col:mse} and Theorem \oldref{thm:pac}. Let $0<q<1$ and define $\alpha = 1- \exp(-\frac{(1-q)\mu_{min}(\epsilon)}{c})$. Then, with probability at least $1-e^{-CK}$ it holds that
	\begin{equation}\label{eq:pac1}
\mathrm{MSE}_{S_k}\leq\alpha \mathrm{MSE}_{o} + (1-\alpha) \mathrm{Tr}(\P_{k|k-1}),
	\end{equation}
	for some $C>0$.
\end{corollary}
\vspace{-0.2cm}
\section{Simulation Results}\label{sec:sim}
To test the performance of the proposed randomized greedy algorithm, we compare it with the classic greedy algorithm and the SDP relaxation in a variety of settings as detailed next. We implemented the greedy and randomized greedy algorithms in MATLAB and the SDP relaxation scheme via CVX \cite{grant2008cvx}. All experiments were run on a laptop with 2.0 GHz Intel Core i7-4510U CPU and 8.00 GB of RAM.
\vspace{-0.2cm}
\subsection{Kalman filtering in random sensor networks}
We first consider the problem of state estimation in a linear time-varying system via Kalman filtering. 
For simplicity, we assume the state transition matrix to be identity, i.e., $\mathbf{A}_k = \mathbf{I}_m$.
At each time step, 
the measurement vectors, i.e., the rows of the measurement matrix $\H_k$, are drawn according to 
$\mathcal{N} \sim (0,\frac{1}{m}\I_m)$. The
initial state is a zero-mean Gaussian random vector with covariance $\mathbf{\Sigma_x} = \I_m$; and the process and measurement noise are zero-mean 
Gaussian with covariance matrices $\mathbf{Q}=0.05\I_m$ and $\mathbf{R}=0.05\I_n$, respectively.

The MSE of the filtered estimator and running time of each scheme is averaged over 100 Monte-Carlo simulations. The 
time horizon for each run is $T = 10$. 

We first consider a system having state dimension $m = 50$ and the total number of sensors $n = 400$. We set a
constraint on the number of sensors allowed to be queried at each time step to $K = 55$ and compare the MSE 
achieved by each sensor selection method over the time horizon of interest. For the randomized greedy 
algorithm we set $\epsilon = 0.001$. Fig. \oldref{fig:overtime} shows that the greedy method consistently yields the 
lowest estimation MSE while the MSE provided by the randomized greedy algorithm is slightly higher. The MSE 
performance achieved by solving the SDP relaxation is considerably larger than those of the greedy and randomized 
greedy algorithms. The time it takes each method to select $K$ sensors is given in Table \oldref{tab:sim1_time}. 
Both the greedy algorithm and the randomized greedy algorithm are much faster than the SDP formulation.
Moreover, the randomized greedy scheme is nearly two times faster than the greedy method.
\begin{table}[b]
\centering
\begin{tabular}{|c|c|c|}
\hline
Randomized Greedy & Greedy & SDP Relaxation \\ 
\hline
0.20 s & 0.38 s & 249.86 s\\
\hline
\end{tabular}
\caption{Running time comparison of the randomized greedy, greedy, and SDP relaxation sensor selection schemes 
($m=50$, $n=400$, $K=55$, $\epsilon=0.001$).}
\vspace{-0.4cm}
\label{tab:sim1_time}
\end{table}
\begin{figure}[t]
\centering
    \includegraphics[width=0.49\textwidth]{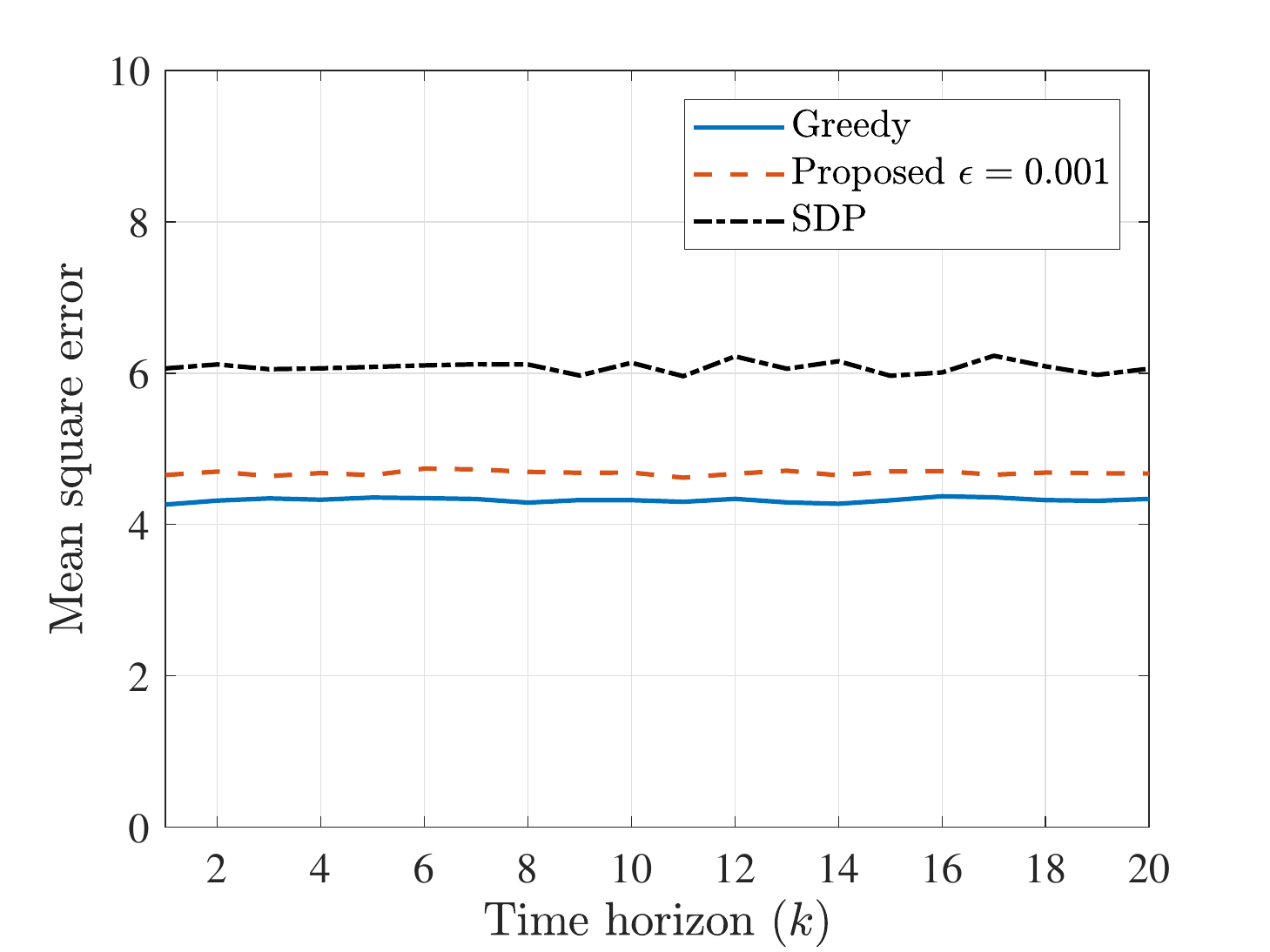}
    \caption{MSE comparison of randomized greedy, greedy, and SDP relaxation sensor selection schemes employed in Kalman filtering.}
\label{fig:overtime}
\vspace{-0.3cm}
\end{figure}
Note that, in this example, in each iteration of the sensor selection procedure the randomized scheme only computes 
the marginal gain for a sampled subset of size 50. In contrast, the classic greedy approach computes the marginal gain 
for all 400 sensors. In summary, the greedy method yields slightly lower MSE but is much slower than the proposed 
randomized greedy algorithm.

To study the effect of the number of selected sensors on the MSE performance, we vary $K$ from 55 to 115 with 
increments of 10. The MSE values at the last time step for each algorithm are shown in Fig. \oldref{fig:kvary}(a). 
As the number of selected sensors increases, the estimation becomes more accurate, as reflected by the MSE 
of the estimates provided by each algorithm. Moreover, the differences between the MSE values achieved by
different schemes monotonically decrease as more sensors are selected. The sensor selection running times  shown in 
Fig. \oldref{fig:kvary}(b) indicate that the randomized greedy scheme is nearly twice as fast as the greedy method, 
while the SDP method is orders of magnitude slower than both greedy and randomized greedy algorithms. 

\begin{figure}[t]
\centering
    \begin{subfigure}{.49\textwidth}
  \centering
    \includegraphics[width=1\textwidth]{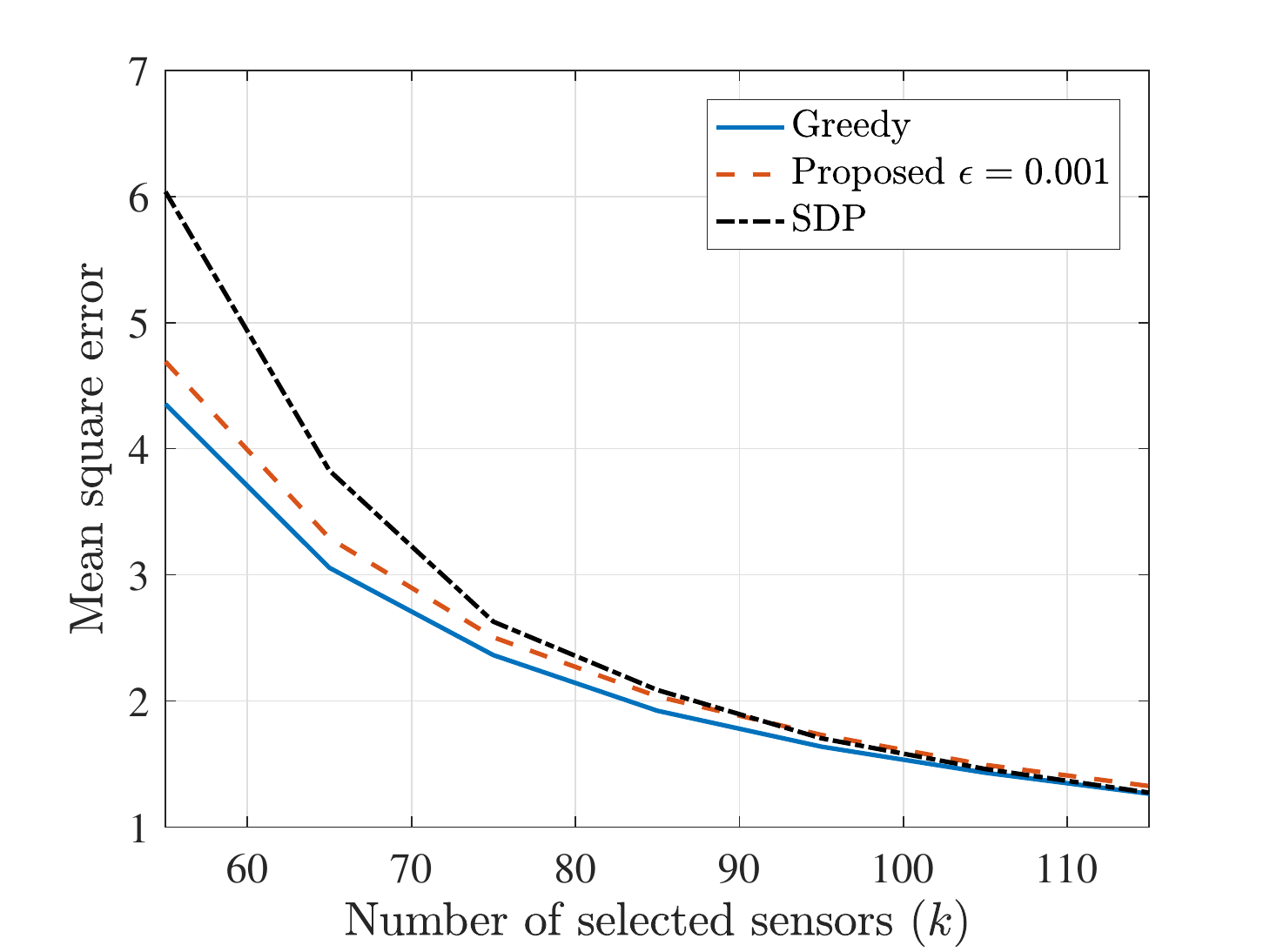}\caption{\scriptsize Comparing MSE performance of different 
    schemes.}
        \end{subfigure}
        \begin{subfigure}{.49\textwidth}
  \centering
    \includegraphics[width=1\textwidth]{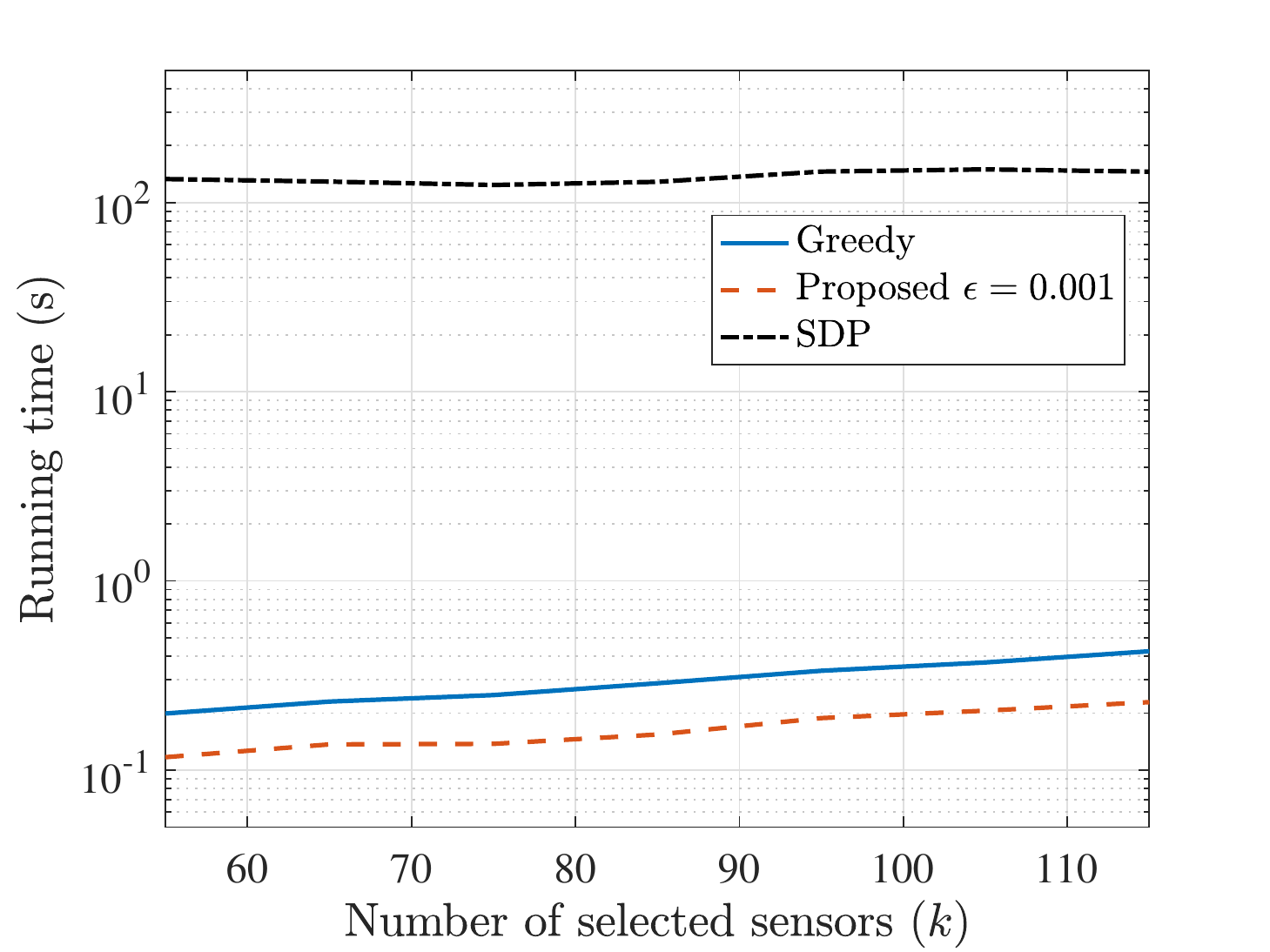}\caption{\scriptsize Running time comparison.}
\end{subfigure}
\caption{Comparison of randomized greedy, greedy, and SDP relaxation schemes as the number of selected sensors increases.}
\label{fig:kvary}
\vspace{-0.4cm}
\end{figure}
Finally, to empirically verify the results of Theorem \oldref{thm:pac}, in Fig. \oldref{fig:hists} we compare histograms of 
MSE achieved by the greedy and the proposed randomized greedy sensor selection schemes with various choices of 
$\epsilon$ when $K = 60$. As the figure shows, the MSE of sets selected by the proposed scheme is relatively close 
to that selected by state-of-the-art greedy algorithm. In addition, as $\epsilon$ decreases, the MSE of the randomized 
greedy algorithm approaches that of the greedy algorithm. These empirical observations coincide with our theoretical 
results in Theorem \oldref{thm:pac}. That is, the proposed algorithm, although a randomized scheme, returns a 
near-optimal subset of sensors for each individual sensor selection task.
\begin{figure}[t]
	\centering
	\includegraphics[width=0.49\textwidth]{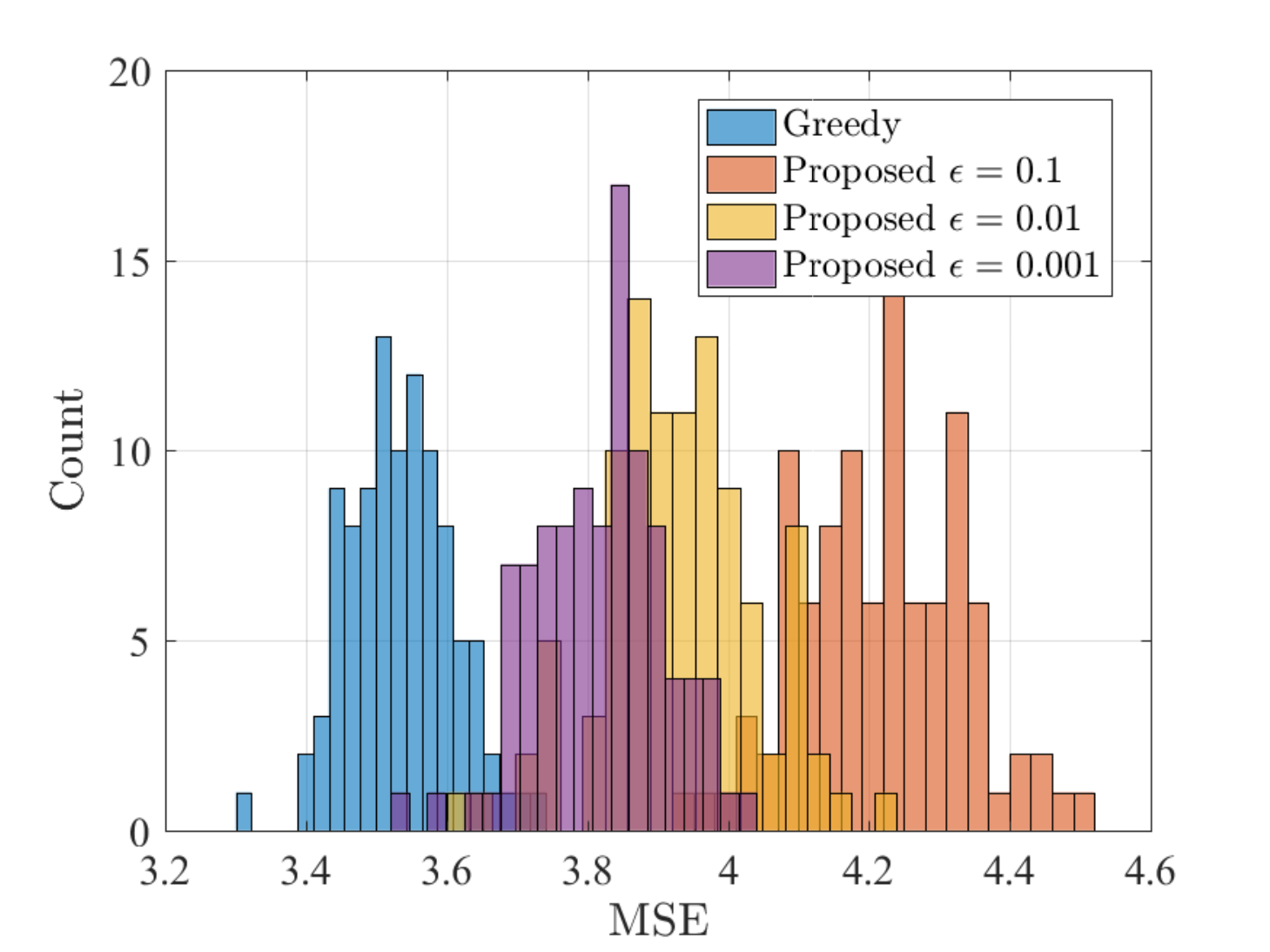}
	\caption{Histogram of MSE values for 100 independent realization of a sensor scheduling task for a sensor network with $m = 50$, $K=60$, and $n = 400$.}
	\label{fig:hists}
	\vspace{-0.4cm}
\end{figure}
\vspace{-0.2cm}
\subsection{State estimation in large-scale networks}
Next, we compare the performance of the randomized greedy algorithm to that of the greedy algorithm as the size 
of the system increases. We run both methods for 20 different system dimensions. The initial dimensions are set to 
$m = 20$, $n=200$, and $K=25$ and all three parameters are scaled by $\gamma$ where $\gamma$ varies from 1 to 20. 
In addition, to evaluate the effect of $\epsilon$ on the performance and runtime of the randomized greedy approach, 
we repeat experiments for $\epsilon \in \{0.1, 0.01, 0.001\}$. Note that the computational complexity of the SDP 
relaxation scheme is prohibitive in this setting and hence it is omitted. Fig. \oldref{fig:gammavary}(a) illustrates the 
MSE comparison of the greedy and randomized greedy schemes.
It shows that the 
difference between the MSEs is negligible.
 The running time is plotted in 
Fig. \oldref{fig:gammavary}(b). As the figure illustrates, the gap between the running times grows with the size of the 
system and the randomized greedy algorithm performs nearly 28 times faster than the greedy method for the largest 
network. Fig. \oldref{fig:gammavary} shows that using a smaller $\epsilon$ results in a lower MSE  while it slightly 
increases the running time. These results suggest that, for large systems, the randomized greedy provides almost 
the same MSE while being much faster than the greedy algorithm.
\begin{figure}[t]
\centering
    \begin{subfigure}{.49\textwidth}
  \centering
    \includegraphics[width=1\textwidth]{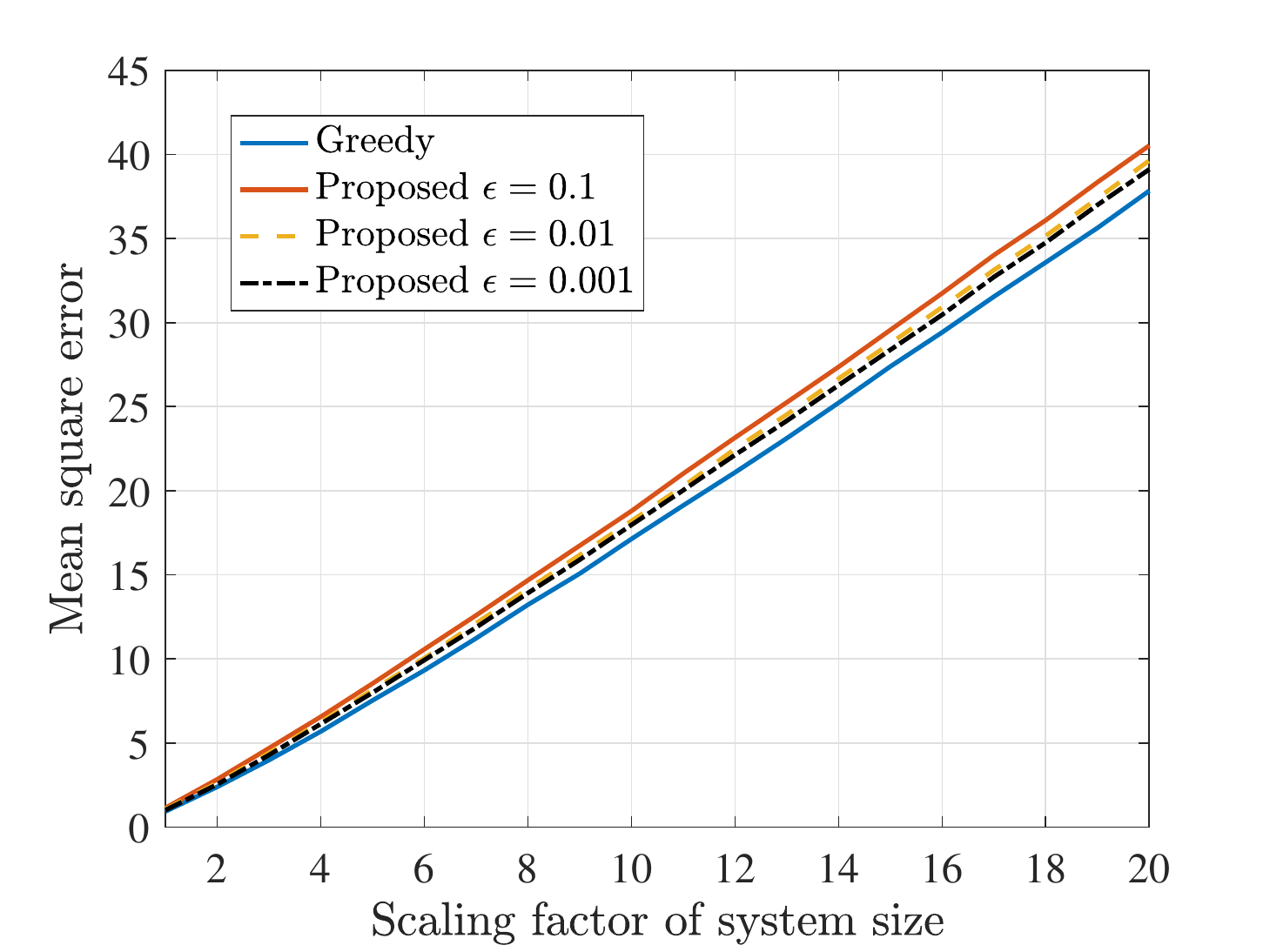}\caption{\scriptsize {\color{black}{Comparing MSE performance of different 
    			schemes.}}}
        \end{subfigure}
        \begin{subfigure}{.49\textwidth}
  \centering
    \includegraphics[width=1\textwidth]{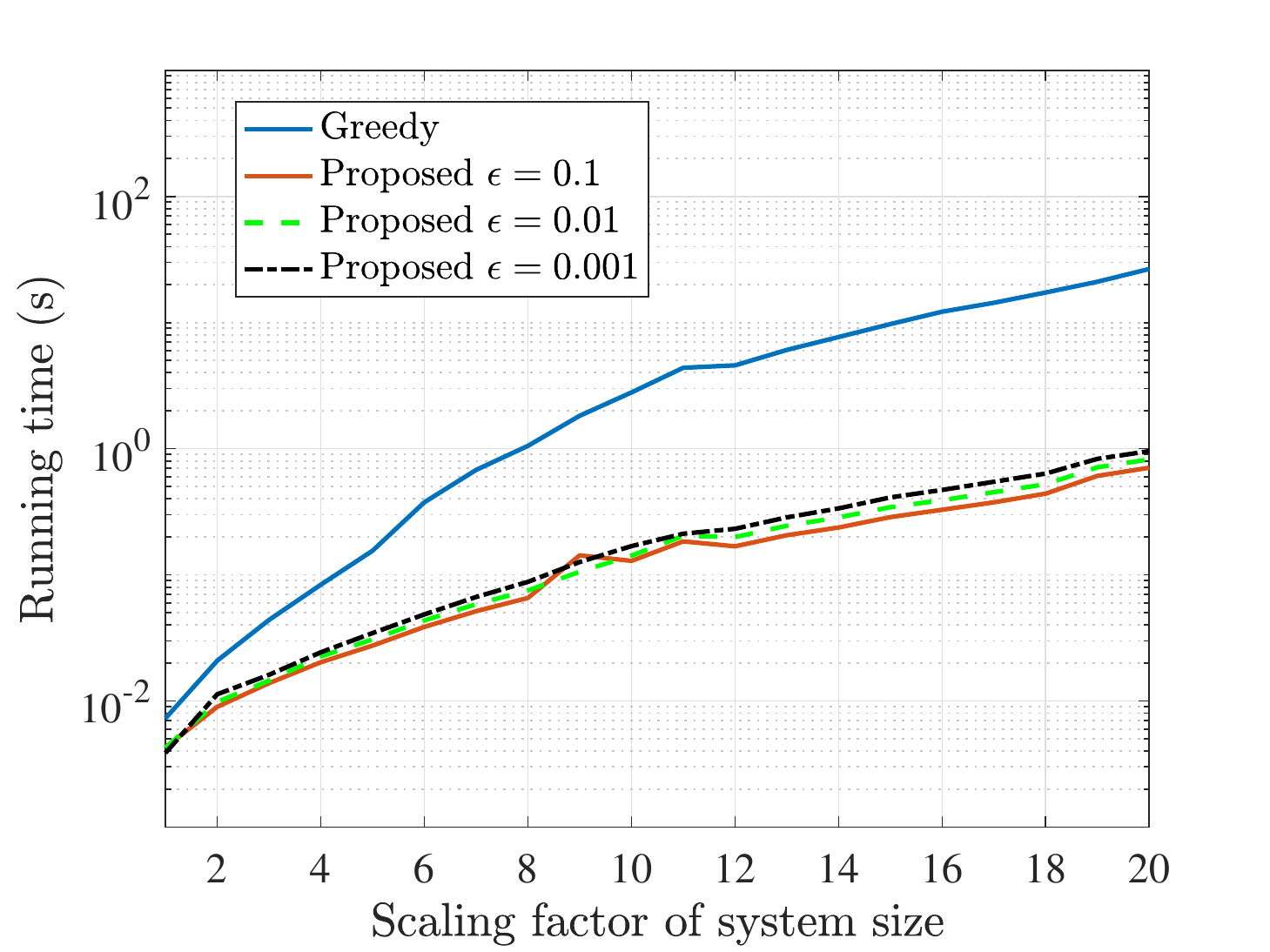}\caption{\scriptsize Running time comparison.}
\end{subfigure}
\caption{A comparison of the randomized greedy and greedy algorithms for varied network size.}
\label{fig:gammavary}
\vspace{-0.4cm}
\end{figure}
\vspace{-0.2cm}
\subsection{Accelerated multi-object tracking}
\begin{figure}[t]
	\centering
	\begin{subfigure}{.49\textwidth}
		\centering
		\includegraphics[width=1\textwidth]{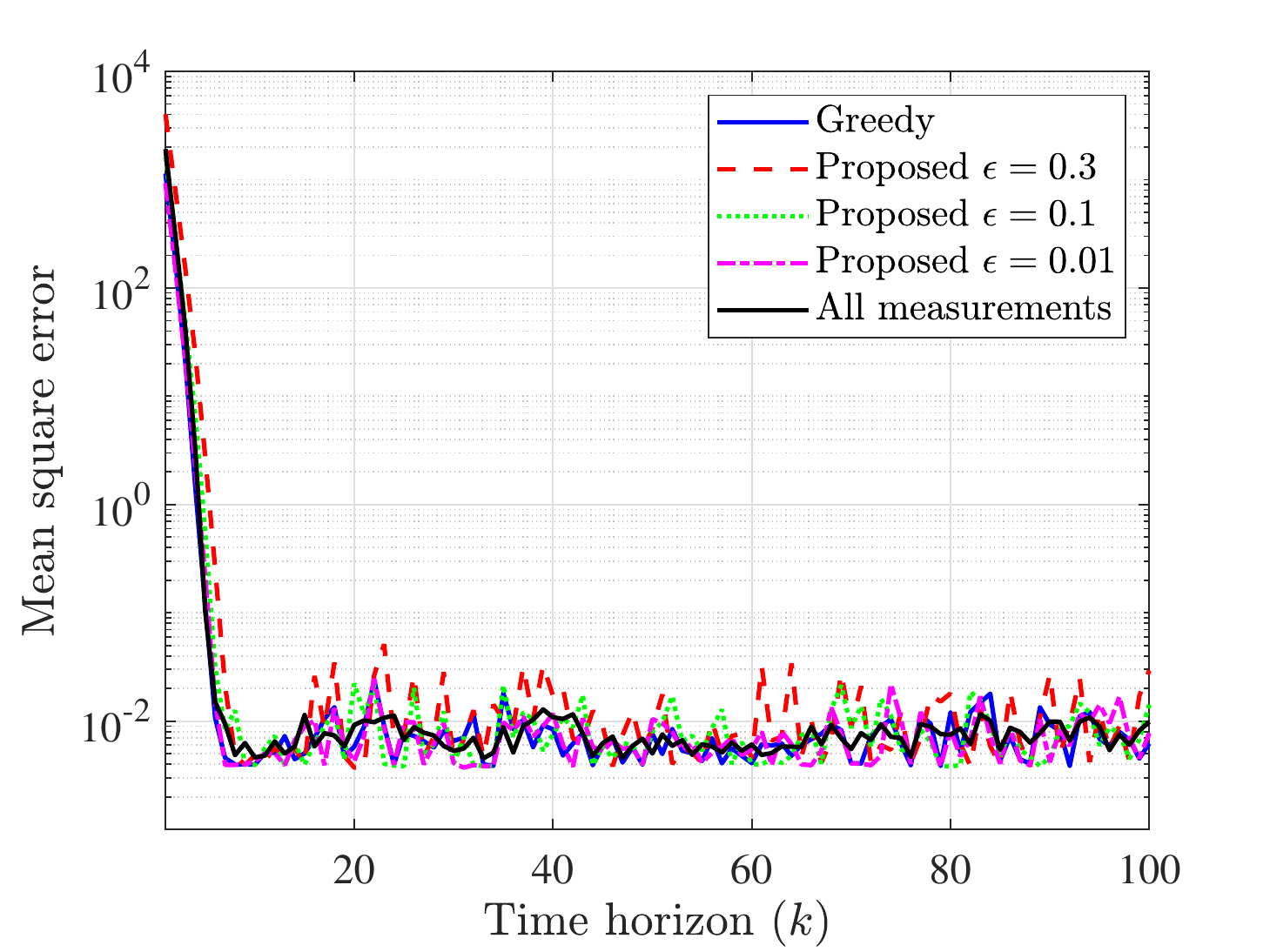}\caption{\scriptsize Comparing MSE performance of different 
			schemes.}
	\end{subfigure}
	\begin{subfigure}{.49\textwidth}
		\centering
		\includegraphics[width=1\textwidth]{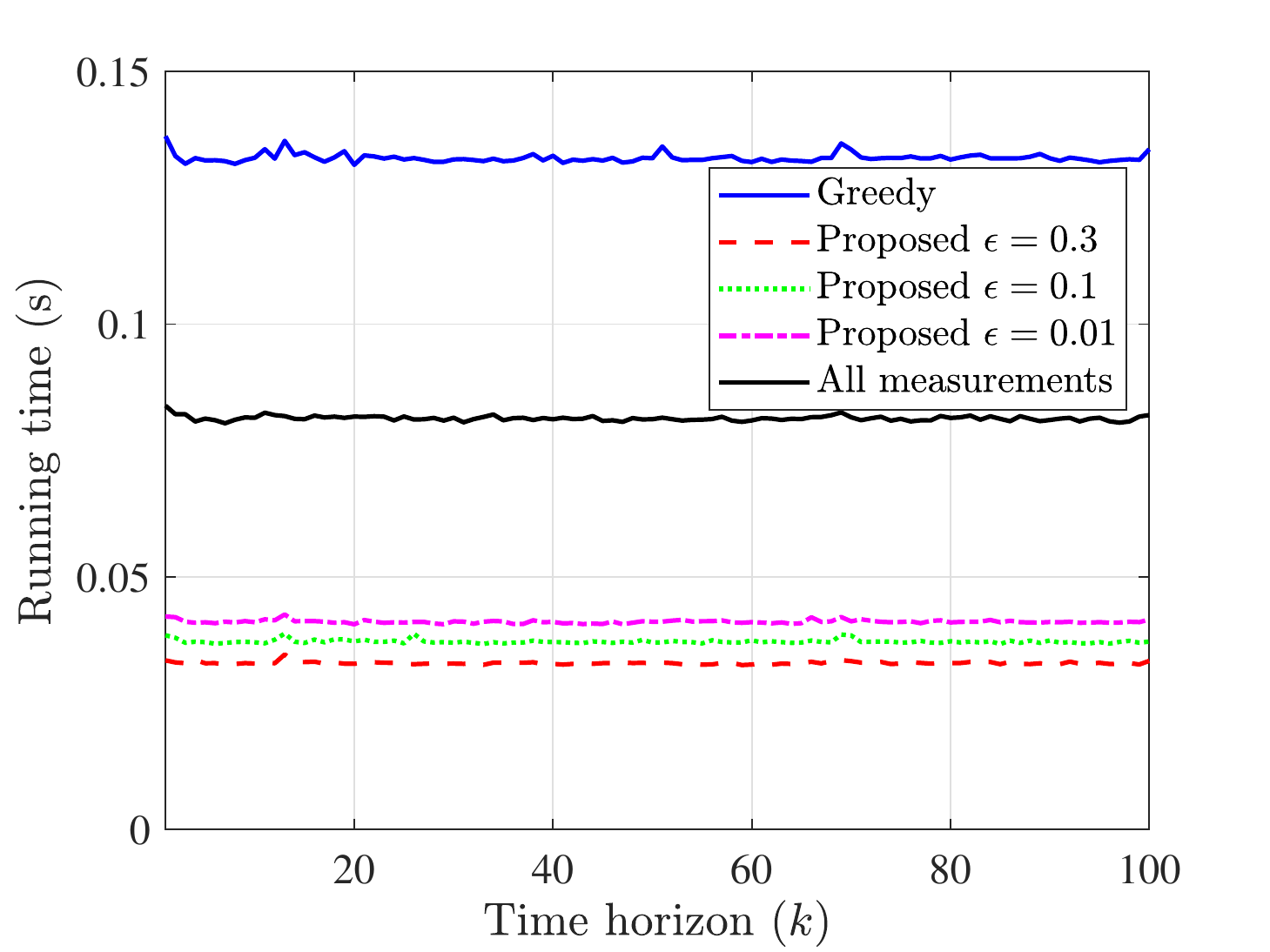}\caption{\scriptsize Running time comparison.}
	\end{subfigure}
\caption{A comparison of the randomized greedy and greedy algorithms for a multi-object tracking application.}
\label{fig:uavs}
\vspace{-0.4cm}
\end{figure}
Finally, we study the multi-object tracking application introduced in Section \oldref{sec:uav}. 
Specifically, we consider a scenario where twenty moving objects are initially uniformly distributed in a $5\times 10$ area. At each time instance, the objects move in a random direction with a constant velocity set to 0.2. 
Twenty UAVs, equidistantly spread over the area, move according to a periodic \textit{parallel-path} search pattern \cite{vincent2004framework}. 
The initial phases of the UAVs' motions are uniformly distributed to provide a better coverage of the area. 
The UAVs can acquire range and angular measurements of the objects that are within the maximum radar detection range. 
The maximum radar detection range is set such that at each time step the UAVs together collect approximately 600 range 
and angular measurements. The communication bandwidth constraints limit the number of measurements
transmitted to the control unit to $K = 100$. Note that since the radar measurement model is nonlinear, the 
control unit tracks objects via the extended Kalman filter. 
Fig. \oldref{fig:uavs} shows a comparison in terms of the MSE and running time between the greedy and randomized greedy schemes for various values of $\epsilon$.
In the same figure we show performance of the scheme that ignores communications constraints 
and uses all the available measurements gathered by the UAVs. As Fig. \oldref{fig:uavs}(a) illustrates, the 
MSE performance of the greedy and proposed schemes are relatively close and similar to the performance
of the scheme that uses all the measurements. However, a closer look at the running time comparison shown in 
Fig. \oldref{fig:uavs}(b) reveals that the combined runtime of randomized greedy sensor selection and Kalman 
filtering tasks is approximately $2$ times faster than the runtime of the Kalman filter that uses all the 
measurements, and approximately $4$ times faster than the combined runtime of the classical greedy sensor
selection and Kalman filtering. Therefore, not only does the proposed scheme satisfy the communication 
constraint and perform nearly as well as using all the measurements, but it also significantly reduces the time 
needed to perform sensor selection and process the selected measurements in extended Kalman filtering.
\vspace{-0.2cm}
\section{Conclusion} \label{sec:concl}
In this paper, we studied the problem of state estimation in large-scale linear time-varying dynamical systems. 
We proposed a randomized greedy algorithm for selecting sensors to query such that their choice minimizes the 
estimator's mean-square error at each time step. We established the performance guarantee for the proposed 
algorithm and analyzed its computational complexity. To our knowledge, the proposed scheme is the first 
randomized algorithm for sensor scheduling with an explicit bound on its achievable mean-square error.
In addition, we provided a probabilistic theoretical bound on the element-wise curvature of the objective function. 
Furthermore, in several simulated settings we demonstrated that the proposed algorithm is superior to the classical 
greedy and SDP relaxation methods in terms of running time while providing the same or better utility.

As a future work, it is of interest to extend this approach to nonlinear dynamical systems and obtain theoretical 
guarantees on the quality of the resulting approximate solution found by the randomized greedy algorithm. 
Moreover, it would be of interest to extend the framework established in this manuscript to related problems 
such as the minimal actuator placement.
\normalsize
\vspace{-0.2cm}
\begin{appendices}
\vspace{-0.3cm}
\section{Proof of Proposition \oldref{p:mono}}\label{pf:mono}
First, note that 
\begin{equation}
f(\emptyset)=\mathrm{Tr}\left(\P_{k|k-1}-\F_{\emptyset}^{-1}\right)=\mathrm{Tr}\left(\P_{k|k-1}-\P_{k|k-1}\right)=0.
\end{equation}
Now, for $j\in[n]\backslash S$ it holds that
\begin{equation}\label{eq:update}
\begin{aligned}
f_j(S) &= f(S\cup \{j\}) -f(S) \\
&= \mathrm{Tr}\left(\P_{k|k-1}-\F_{S\cup \{j\}}^{-1}\right) -  \mathrm{Tr}\left(\P_{k|k-1}-\F_S^{-1}\right)\\
&= \mathrm{Tr}\left(\F_S^{-1}\right) -\mathrm{Tr}\left(\F_{S\cup \{j\}}^{-1}\right)\\
&=  \mathrm{Tr}\left(\F_S^{-1}\right)-\mathrm{Tr}\left(\left(\F_S+\sigma_j^{-2}\h_{k,j}\h_{k,j}^\top\right)^{-1}\right)\\
&\stackrel{(a)}{=} 
\mathrm{Tr}\left(\frac{\F_S^{-1}\h_{k,j}\h_{k,j}^\top\F_S^{-1}}{\sigma_j^{2}+\h_{k,j}^\top\F_S^{-1}\h_{k,j}}\right)  \\
&\stackrel{(b)}{=} \frac{\h_{k,j}^\top\F_S^{-2}\h_{k,j}}{\sigma_j^{2}+\h_{k,j}^\top\F_S^{-1}\h_{k,j}}
\end{aligned}
\end{equation}
where $(a)$ is obtained by applying matrix inversion lemma (Sherman-Morrison formula) 
\cite{bellman1997introduction} to $(\F_S+\sigma_j^{-2}\h_{k,j}\h_{k,j}^\top)^{-1}$, and $(b)$ follows from the 
properties of the matrix trace operator. Finally, since $\F_S$ is a symmetric positive definite matrix, 
$f_j(S) > 0$ which in turn implies monotonicity.
\vspace{-0.4cm}
\section{Proof of Lemma \oldref{lem:curv}}\label{pf:curv}
Let $S \subset T$ and $T \backslash S=\{j_1,\dots,j_r\}$. Therefore,
\begin{multline}
f(T)-f(S)=f(S\cup \{j_1,\dots,j_r\})-f(S)\\
\hspace{0.5cm}={f_{j_1}(S)}+{f_{j_2}(S\cup\{j_1\})}+\dots\\+{f_{j_r}(S\cup\{j_1,\dots,j_{r-1}\})}.
\end{multline}
Definition of the element-wise curvature implies that
\begin{equation}\label{ds}
\begin{aligned}
f(T)-f(S)&\leq {f_{j_1}(S)}+{\cal C}_1{f_{j_2}(S)}+\dots+{\cal C}_{r-1}{f_{j_r}(S)}\\
&={f_{j_1}(S)}+\sum_{l=1}^{r-1}{\cal C}_l{f_{j_t}(S)}.
\end{aligned}
\end{equation}
Note that \ref{ds} is established for a specific ordering of elements in $T \backslash S$. Given 
an ordering $\{j_1,\dots,j_r\}$, one can form a set $P=\{{\cal P}_1,\dots,{\cal P}_r\}$ of $r$ permutations 
(e.g., by defining the right circular-shift operator ${\cal P}_t(\{j_1,\dots,j_r\})=\{j_{r-t+1},\dots,j_1,\dots\}$ for 
$1 \leq t\leq r$) such that ${\cal P}_p(j)\neq{\cal P}_q(j)$ for $p\neq q$ and $\forall j\in T \backslash S$;
\ref{ds} holds for each such permutation. By summing the corresponding $r$ inequalities we obtain
\begin{equation}\label{qs}
r(f(T)-f(S))\leq \left(1+\sum_{l=1}^{r-1}{\cal C}_l\right)\sum_{j\in T \backslash S}{f_{j}(S)}.
\end{equation}
Rearranging \ref{qs} yields the desired result.
\vspace{-0.4cm}
\section{Proof of Lemma \oldref{lem:rand}}\label{pf:rand}
First, we aim to bound the probability of an event that a random set $R$ contains at least one index from 
the optimal set of sensors which is a necessary condition to reach the optimal MSE. Let us consider $S_t^{(i)}$, 
the set of sensors selected by the end of $i\ts{th}$ iteration of Algorithm \oldref{alg:greedy} and let 
$\Phi = R \cap (O_k\backslash S_t^{(i)})$. It holds 
that\footnote{Without a loss of generality, we assume that $s$ is an integer.}
\begin{equation}
\begin{aligned}
\Pr\{\Phi = \emptyset\}& = \prod_{l = 0}^{s-1} \left(1-\frac{|O_k\backslash S_k^{(i)}|}{|[n]\backslash S_k^{(i)}|-l}\right)\\
&\stackrel{(a)}{\leq}\left(1-\frac{|O_k\backslash S_k^{(i)}|}{s}\sum_{l=0}^{s-1}\frac{1}{|[n]\backslash S_k^{(i)}|-l}\right)^s\\
&\stackrel{(b)}{\leq}\left(1-\frac{|O_k\backslash S_k^{(i)}|}{s}\sum_{l=0}^{s-1}\frac{1}{n-l}\right)^s
\end{aligned}
\end{equation}
where $(a)$ holds due to the inequality of arithmetic and geometric means, and $(b)$ holds since 
$|[n]\backslash S_i|\leq n$. Now recall that for any integer $p$,
\begin{equation}\label{eq:har}
H_p=\sum_{l=1}^p\frac{1}{p}=\log p + \gamma+\zeta_p,
\end{equation}
where $H_p$ is the $p\ts{th}$ harmonic number, $\gamma$ is the Euler-Mascheroni constant, and 
$\zeta_p	 = \frac{1}{2p} - \mathcal{O_k}(\frac{1}{p^4})$ is a monotonically decreasing sequence 
related to Hurwitz zeta function \cite{lang2013algebraic}. Therefore, using the identity \ref{eq:har} 
we obtain
\begin{equation}
\begin{aligned}
\Pr\{\Phi = \emptyset\}&\stackrel{}{\leq} (1-\frac{|O_k\backslash S_k^{(i)}|}{s}(H_n-H_{n-s}))^s \\
&\stackrel{}{=} (1-\frac{|O_k\backslash S_k^{(i)}|}{s}(\log(\frac{n}{n-s})+\zeta_n-\zeta_{n-s}))^s\\
&\stackrel{(c)}{\leq} (1-\frac{|O_k\backslash S_k^{(i)}|}{s}(\log(\frac{n}{n-s})-\frac{s}{2n(n-s)}))^s\\
&\stackrel{(d)}{\leq} ((1-\frac{s}{n})e^{\frac{s}{2n(n-s)}})^{|O_k\backslash S_k^{(i)}|},
\end{aligned}
\end{equation}
where $(c)$ follows since $\zeta_n-\zeta_{n-s} = \frac{1}{2n} -\frac{1}{2(n-s)} + 
\mathcal{O_k}(\frac{1}{(n-s)^4})$, and $(d)$ is due to the fact that $(1+x)^y\leq e^{xy}$ for any 
real number $y>0$. Next, the fact that $\log (1-x)\leq -x-\frac{x^2}{2}$ for $0<x<1$ yields
\begin{equation}
(1-\frac{s}{n})e^{\frac{s}{2n(n-s)}} \leq e^{-\frac{\beta_1 s}{n}},
\end{equation}
where $\beta_1 = 1 +(\frac{s}{2n}-\frac{1}{2(n-s)})$. On the other hand, we can also upper-bound $\Pr\{\Phi = \emptyset\}$ as
\begin{equation}
\begin{aligned}
\Pr\{\Phi = \emptyset\}& \leq \left(1-\frac{|O_k\backslash S_k^{(i)}|}{s}\sum_{l=0}^{s-1}\frac{1}{n-l}\right)^s\\
& \leq\left(1-\frac{|O_k\backslash S_k^{(i)}|}{n}\right)^s\\
&\leq e^{-\frac{s}{n}|O_k\backslash S_k^{(i)}|},
\end{aligned}
\end{equation}
where we again employed the inequality  $(1+x)^y\leq e^{xy}$. Let us denote $\beta = \max\{1,\beta_1\}$. 
Then 
\begin{equation}\label{eq:pbound}
\Pr\{\Phi \neq \emptyset\} \geq 1- e^{-\frac{\beta s}{n}|O_k\backslash S_k^{(i)}|}\geq \frac{1-\epsilon^\beta}{K}(|O_k\backslash S_k^{(i)}|),
\end{equation}
by the definition of $s$ and the fact that $1- e^{-\frac{\beta s}{n}x} $ is a concave function. Finally, according to 
Lemma 2 in \cite{mirzasoleiman2014lazier},
\begin{equation}\label{eq:lazy}
\E[f_{(i+1)_s}(S_k^{(i)})|S_k^{(i)}]\geq \frac{\Pr\{\Phi \neq \emptyset\}}{|O_k\backslash S_k^{(i)}|}\sum_{j\in O_k\backslash S_k^{(i)}}f_o(S_k^{(i)}).
\end{equation}
Combining \ref{eq:pbound} and \ref{eq:lazy} yields the stated results.
\end{appendices}
\bibliographystyle{IEEEtran}
\bibliography{IEEEabrv,refs}
\end{document}